\documentclass{elsarticle}
\pdfoutput=1
\makeatletter
\def\ps@pprintTitle{
  \let\@oddfoot\@empty
}
\makeatother

\usepackage[margin=1in]{geometry}

\usepackage{bbm}
\usepackage{graphicx}
\usepackage{wrapfig}
\usepackage{caption}
\usepackage{subcaption}
\usepackage[section]{placeins} % Defines \FloatBarrier which does not allow
\usepackage{fancyhdr} % Required for custom headers
\usepackage{lastpage} % Required to determine the last page for the footer
\usepackage{extramarks} % Required for headers and footers

\usepackage{xcolor} % Colors
\usepackage{enumerate} % Additional options for \enumerate{}
\usepackage{paralist} % Additional options for lists; inline lists
\usepackage{amsmath, amsthm, amssymb, mathtools}
\usepackage{mathabx, pifont, stmaryrd} % More math symbols, fonts
\usepackage{etoolbox} % Boolean variables and more
\usepackage{bibentry} % Bibliography entries anywhere in doc
\makeatletter\let\saved@bibitem\@bibitem\makeatother % Do not remove if bibentry
\usepackage[colorlinks, bookmarksopen, bookmarksnumbered,
            citecolor=red,urlcolor=red]{hyperref} % Hyperlinks
\makeatletter\let\@bibitem\saved@bibitem\makeatother % Do not remove if bibentry
\usepackage{comment} % Allows multiline comments

\usepackage{multirow}

\usepackage{algorithm}
\usepackage{algorithmic}

\usepackage{optidef}
\usepackage{mathtools}

\newtheorem{lemma}{Lemma}
\newtheorem{prop}{Proposition}
\newtheorem{remark}{Remark}

\theoremstyle{definition}

\usepackage{bm} % required for: \bm{}
\usepackage{amsfonts} % required for: \mathbb{}, \mathcal{}, ...

\DeclareMathOperator*{\argmin}{arg\,min}

\usepackage{tikz}
\usepackage{pgfplots}
\usepackage{xcolor}
\usepackage{pgfplotstable, booktabs}
\pgfplotsset{compat=1.9}
\definecolor{lightergreen}{rgb}{0.8, 1.0, 0.8}
\definecolor{seagreen}{rgb}{0.180, 0.545, 0.341}

\usetikzlibrary{pgfplots.groupplots}
\usepgfplotslibrary{fillbetween}
\usetikzlibrary{calc,fit,matrix,arrows,automata,positioning,shapes}
\usetikzlibrary{arrows.meta}

\pgfplotsset{select coords between index/.style 2 args={
    x filter/.code={
        \ifnum\coordindex<#1\def\pgfmathresult{}\fi
        \ifnum\coordindex>#2\def\pgfmathresult{}\fi
    }
}}

\tikzset{
 invisible/.style={opacity=0},
 visible on/.style={alt={#1{}{invisible}}},
 alt/.code args={<#1>#2#3}{%
   \alt<#1>{\pgfkeysalso{#2}}{\pgfkeysalso{#3}}
 },
}

\usepackage{enumitem}

\begin{document}
\title{Feature tracking in physics-informed neural networks via joint
       optimization of nonlinear deformation manifolds:
       application to shocks}

\author[rvt2]{Akshay Thakur\fnref{fn1}\corref{cor1}}
\ead{athakur3@nd.edu}

\author[rvt2]{Matthew J. Zahr\fnref{fn2}}
\ead{mzahr@nd.edu}

\address[rvt2]{Department of Aerospace and Mechanical Engineering, University
               of Notre Dame, Notre Dame, IN 46556, United States}

\cortext[cor1]{Corresponding author}

\fntext[fn1]{Graduate Student, Department of Aerospace and Mechanical
             Engineering, University of Notre Dame}
\fntext[fn2]{Associate Professor, Department of Aerospace and Mechanical
             Engineering, University of Notre Dame}


\begin{abstract}
Physics-informed neural networks (PINNs) can converge to inaccurate solutions for conservation laws whose solutions develop shocks and other localized features.
 A uniformly distributed set of collocation points undersamples such features while oversampling the remainder of the domain, causing the residual to be dominated 
 by regions in which the solution is already well approximated. We propose a feature-tracking PINN (FT-PINN) framework in which the solution network is defined on
  a fixed reference domain and composed with a diffeomorphic deformation map drawn from a parameterized nonlinear manifold of deformations. The deformation
   parameters are trained jointly with the solution-network parameters by minimizing the pulled-back conservation-law residual. The nonlinear deformation map 
   enables the collocation points to concentrate along features of essentially arbitrary geometry, including curved, oblique, and merging shocks. Because the 
   deformation is determined by the optimization, the method requires no a priori knowledge of the feature location to concentrate collocation points in its 
   vicinity. The framework is agnostic to the particular nonlinear parameterization used for the deformation map. However, nonlinear parameterizations do not, 
   in general, guarantee invertibility of the resulting map or preservation of the domain geometry. Boundary preservation is therefore imposed exactly
   through a tangential projection of the displacement, while folding of the map is discouraged by a one-sided penalty on the Jacobian determinant
   evaluated at the collocation points.
   We assess the framework on four problems involving shocks as localized features, namely the space-time viscous Burgers equation with two merging shocks,
   a decelerating Burgers shock, the one-dimensional Euler shock tube in space-time, and steady two-dimensional Euler regular shock reflection from a wall. 
   In all four cases, FT-PINN resolves the shocks at their correct locations using a limited collocation budget, whereas a vanilla PINN with an identical
    network architecture, collocation budget, and training procedure either misplaces the shocks or fails to develop them altogether.

\end{abstract}

\maketitle

\section{Introduction}
\label{sec:intro}
% !TEX root = main.tex
Physics-informed neural networks
(PINNs)~\cite{raissi2019physics,karniadakis2021physics} approximate solutions of
partial differential equations (PDEs) by training a neural network to minimize a
composite loss function. The loss comprises the residual of the governing PDE,
evaluated at a set of collocation points distributed over the computational
domain, together with penalty terms enforcing the boundary and initial
conditions. PINNs have been successfully employed in a wide range of
scientific and engineering applications~\cite{wang2023expert,wang2024respecting, wang2024piratenets,navaneeth2024physics}, including fluid mechanics~\cite{cai2021physics,mao2020physics}, solid mechanics~\cite{haghighat2021physics}, 
inverse problems~\cite{yang2021b,raissi2019physics}, and surrogate modeling~\cite{zhu2019physics,haghighat2021physics}. 

\subsection{Resolving shocks using PINNs}
\label{subsec:resolveshock}

Despite their flexibility, PINNs trained on a uniformly distributed set of
collocation points are known to converge to solutions that are inaccurate,
smeared over a wide band, or missing the feature altogether on problems whose
solutions develop shocks and other localized
features~\cite{mao2020physics,fuks2020limitations,arzani2023theory}. First, a
uniform collocation set is a poor discretization of the residual. The smooth
regions occupying most of the domain contribute nearly all of the loss and bias
the optimization toward polishing them, while the feature---narrow relative to
the point spacing---is sampled by too few points to be resolved.

Second, and more fundamentally, a hyperbolic conservation law admits genuinely
discontinuous weak solutions, so the strong form residual that drives the
physics-informed loss is not valid at the shock. For the Euler equations, the
differential form is derived under an assumption of smoothness, and consequently,
conserves entropy along streamlines, whereas the admissible solution must jump in
entropy across a shock \cite{wagenaar2023physics}. Minimizing this residual therefore causes the physics-informed network to
seek an isentropic solution to a problem that admits none, which results in an
inaccurate solution \cite{wagenaar2023physics,fuks2020limitations}. The standard remedy is to solve a nearby viscous problem
instead~\cite{LeVeque_2002}. A small artificial viscosity
$\nu$ regularizes the discontinuity into a thin internal layer that supplies the
required dissipation and recovers the entropy solution as $\nu\to0$, and this has been shown to improve solution quality for PINNs~\cite{coutinho2023physics,wagenaar2023physics,fuks2020limitations}.

A number of PINN-based remedies for this sampling deficiency have been proposed.
These include the clustering of collocation points in the vicinity of the
localized feature~\cite{mao2020physics}, the decomposition of the domain into
subregions endowed with separate
subnetworks~\cite{jagtap2020conservative,jagtap2020extended,moseley2023finite},
and error-metric-based adaptation of the point set through resampling or
gradient-based relocation of individual
points~\cite{nabian2021efficient,wu2023comprehensive,lau2023pinnacle,pacmann2024}.
Each is subject to its own limitations. The first two presuppose knowledge of the
feature location prior to training and introduce additional hyperparameters
governing the extent of the clustering or the placement and coupling of the
subdomain interfaces. The third dispenses with such prior knowledge, but incurs
additional cost during training from repeated resampling or point-wise
optimization, and represents the evolution of the collocation set as a collection
of independent per-point trajectories rather than as a single coherent
transformation of the domain.

An alternative, and the approach adopted in this work,
is to introduce a coordinate transformation that clusters the collocation points
at the feature in physical space through the transformation itself rather than through the sampling.
This strategy is well established in mesh-based discretizations and forms the basis of Arbitrary 
Lagrangian--Eulerian formulations~\cite{donea2004arbitrary,persson2009dgdeform}, 
 anisotropic mesh adaptation~\cite{huang2010adaptive}, and implicit shock tracking~
 \cite{zahr2018optimization,zahr2020implicit,huang2022robust,corrigan2019moving,dong2025rp}.
In the latter, $r$-adaptivity is performed on a finite element discretization by solving
an optimization problem that aligns element faces with solution discontinuities,
leaving the discrete solution smooth within each element.

In earlier work~\cite{THAKUR2026110419}, we employed a related approach for tackling boundary-layer problems 
with PINNs, parameterizing the deformation map as a tensor product of independent one-parameter 
nonlinear mappings along each Cartesian direction. While effective for coordinate-aligned features,
 this separable construction restricts the resulting clustering to the coordinate directions and 
 therefore cannot readily resolve features with arbitrary geometry, such as curved or 
 oblique shocks.

\subsection{Proposed approach and contributions}
\label{subsec:proposed}

This article presents a feature-tracking PINN (FT-PINN) that transplants the
$r$-adaptive, shock-fitting paradigm of mesh-based discretizations into the
collocation setting. The solution network is posed on a fixed reference domain
$\Omega_0$ and composed with a parametric diffeomorphism
$\mathcal{G}(\,\cdot\,;\boldsymbol{\tau}):\Omega_0\to\Omega$, trained jointly
with the network parameters against the pulled-back conservation-law residual,
so that a uniform reference collocation grid induces a physical point
distribution that concentrates at the shocks. Departing from the separable tensor-product
construction of our earlier work~\cite{THAKUR2026110419}, the deformation map is drawn from a
general nonlinear manifold capable of aligning with oblique, curved, or geometrically
complex features. The formulation itself is agnostic to how this map is parameterized.
The main contributions of this work are:
\begin{enumerate}[leftmargin=1.25em]
\item a shock-tracking PINN formulation in which collocation point clustering is
produced by a non-separable deformation map optimized jointly with the network
parameters, rather than by resampling or per-point optimization, and without
prior knowledge of the feature location,
\item utilization of a one-sided fold barrier on the determinant of the deformation Jacobian,
a penalty evaluated at the collocation points that discourages folding of the map
without opposing the volume contraction that produces the clustering,
\item an exactly imposed boundary-preservation condition, obtained by applying the
tangential projector to the displacement outputs from the nonlinear manifold, so
that $\mathcal{G}(\partial\Omega_0)=\partial\Omega_0$ holds for every
$\boldsymbol{\tau}$ while boundary points remain free to redistribute
along the boundary, and
\item a demonstration on four shock-dominated problems---two space-time viscous
Burgers configurations, the one-dimensional Euler shock tube in space-time, and
steady two-dimensional Euler regular shock reflection---each compared against a
vanilla PINN of identical architecture, training, and collocation budget, and
against an exact or high-resolution reference solution.
\end{enumerate}

\subsection{Paper outline}
The remainder of the paper is organized as follows.
Section~\ref{sec:goveq} introduces the system of viscous conservation laws, its
space-time form, the vanilla-PINN discretization, and the reference-domain
formulation obtained by pulling the conservation law back through a parametric
diffeomorphism $\mathcal{G}(\,\cdot\,;\boldsymbol{\tau})$.
Section~\ref{sec:optim} develops the feature-tracking PINN over this
formulation, defining the nonlinear manifold of deformations and the RBF
parameterization adopted here, the fold barrier that penalizes folding of the
deformation, the boundary-preserving projector, the induced collocation
distribution, and the resulting joint training objective.
Section~\ref{sec:solver} describes the joint training algorithm, including the
Adam-based mini-batch scheme with separate learning rates and the
identity-initialized warm start. Section~\ref{sec:num_exp} demonstrates the
framework on the four numerical experiments described above. Finally,
Section~\ref{sec:conclude} provides the conclusions.

\section{Governing equations and discretization}
\label{sec:goveq}
% !TEX root = main.tex
\subsection{Viscous conservation laws}
\label{subsec:goveq:claw}

Consider a general system of $m$ viscous conservation laws posed on a
domain $\Omega \subset \mathbb{R}^d$,
\begin{equation}
    \nabla \cdot F(U,\nabla U) = S(U,\nabla U), \qquad x \in \Omega,
    \label{eqn:genconslaw}
\end{equation}
where $U:\Omega \to \mathbb{R}^m$ is the vector of conserved variables
at location $x\in\Omega$,
$F : \mathbb{R}^m \times \mathbb{R}^{m \times d} \rightarrow
\mathbb{R}^{m\times d}$ is the flux function, and
$S:\mathbb{R}^m \times \mathbb{R}^{m \times d} \rightarrow
\mathbb{R}^{m}$ is the source term. The flux is split into
inviscid and viscous parts,
\begin{equation}
    F : (W, Q) \mapsto F_I(W) - F_V(W, Q),
\end{equation}
where $F_I:\mathbb{R}^m \to \mathbb{R}^{m\times d}$ is the inviscid
flux, $F_V:\mathbb{R}^m \times \mathbb{R}^{m\times d} \to
\mathbb{R}^{m\times d}$ is the viscous flux, and we assume the viscous
flux is linear in the second argument:
$F_V : (W, Q) \mapsto D(W) : Q$, where $D : \mathbb{R}^m \to \mathbb{R}^{m\times d\times m\times d}$ is the viscous tensor. Under this splitting,
\eqref{eqn:genconslaw} is equivalent to
\begin{equation}
    \nabla \cdot F_I(U) - \nabla \cdot F_V(U,\nabla U) = S(U,\nabla U),
    \qquad x \in \Omega,
    \label{eqn:genconslaw2}
\end{equation}
supplemented with appropriate boundary conditions on $\partial\Omega$.

\label{subsec:goveq:spacetime}%
No separate governing equations are required for unsteady problems, since
\eqref{eqn:genconslaw2} is general enough to admit them once $\Omega$ is read as a
space-time domain. Throughout, $x\in\Omega$ denotes the ambient point of the
domain. For a steady problem all $d$ of its components are spatial and we write
$x=(x_1,\dots,x_d)$. For an unsteady problem in $d-1$ spatial dimensions the
last component is temporal. We write it as $t$ rather than $x_d$, so that
$x=(x_1,\dots,x_{d-1},t)$, and the divergence
in~\eqref{eqn:genconslaw2} is understood to run over all $d$ space-time
directions. In the flux matrices that follow, columns are ordered such that
the first $d-1$ columns are spatial and the last is temporal.
For any sufficiently smooth function $W : \Omega \to \mathbb{R}^m$, we define the strong-form
residual operator associated with~\eqref{eqn:genconslaw2} as
\begin{equation}
    \mathcal{R}[W](x) \coloneqq
    \nabla \cdot F_I(W(x))
    - \nabla \cdot F_V(W(x),\nabla W(x))
    - S(W(x),\nabla W(x)),
    \label{eqn:residual}
\end{equation}
so that~\eqref{eqn:genconslaw2} becomes $\mathcal{R}[U]\equiv 0$ on
$\Omega$. This operator is the quantity whose pointwise values drive
the physics-informed loss.

\subsection{Physics-informed neural network approximation}
\label{subsec:goveq:pinn}

In a physics-informed neural network~\cite{raissi2019physics}, the
solution of~\eqref{eqn:genconslaw2} is approximated by a neural
network (NN) $\hat{U}_\theta : \Omega \to \mathbb{R}^m$ parameterized by
trainable weights $\theta\in\mathbb{R}^{N_\theta}$. The network is trained to
minimize a composite loss penalizing violation of the PDE and of the
boundary conditions. Given a finite set of collocation points
$\mathcal{X}_c = \{x^{(i)}\}_{i=1}^{N_c}\subset\Omega$ and a set of
boundary points $\mathcal{X}_b=\{x_b^{(j)}\}_{j=1}^{N_b}\subset\partial\Omega$,
the standard physics-informed loss function is
\begin{equation}
    \mathcal{L}(\theta)
    = \frac{w_\mathrm{pde}}{N_c}
      \sum_{i=1}^{N_c}
      \Vert\mathcal{R}[\hat{U}_\theta](x^{(i)})\Vert_2^2
    + \frac{w_\mathrm{bc}}{N_b}\sum_{j=1}^{N_b}
      \bigl\Vert\mathcal{B}[\hat U_\theta](x_b^{(j)}) - g(x_b^{(j)})\bigr\Vert_2^2,
    \label{eqn:pinn_loss}
\end{equation}
where $\mathcal{B}$ is a boundary operator with prescribed data $g$ and
$w_\mathrm{pde},w_\mathrm{bc}\ge 0$ are user-chosen penalty parameters. The derivatives required to evaluate
$\mathcal{R}[\hat{U}_\theta]$ in~\eqref{eqn:residual} are obtained by
reverse-mode automatic differentiation through the network. To train the
network, we seek the weights $\theta_* \in \mathbb{R}^{N_\theta}$ that solve the
unconstrained optimization problem
\begin{equation}
    \theta_* \in \argmin_{\theta\in\mathbb{R}^{N_\theta}}
    \mathcal{L}(\theta),
    \label{eqn:pinn_opt}
\end{equation}
which is typically solved by a first-order optimizer such as Adam. In this work,
$\hat{U}_\theta$ is a fully-connected neural network (FCNN) with a
random Fourier feature embedding~\cite{tancik2020fourier} of the input coordinates. The
input dimension for the network matches the space--time dimension $d$ and the output
dimension matches the conservation-law dimension $m$.

\subsection{Reference-domain formulation}
\label{subsec:goveq:refdom}

Feature-tracking frameworks \cite{zahr2020implicit,dong2025rp,huang2022robust,mirhoseini2023model}
are built around a fixed reference domain $\Omega_0\subset\mathbb{R}^d$ and a
parametric family of diffeomorphisms
$\{\mathcal{G}(\,\cdot\,;\boldsymbol{\tau})\}_{\boldsymbol{\tau}\in\mathcal{T}}$
that deform it onto the physical domain,
\begin{equation}
    \mathcal{G}(\,\cdot\,;\boldsymbol{\tau}) :
    \Omega_0 \longrightarrow \Omega,
    \qquad
    \mathcal{G} : X \mapsto \mathcal{G}(X;\boldsymbol{\tau}),
    \label{eqn:dom-map}
\end{equation}
where $\boldsymbol{\tau}\in\mathcal{T}\subset\mathbb{R}^{N_{\boldsymbol\tau}}$ is a
finite-dimensional parameter vector defining a member of the family and $\mathcal{T}$ is the set of
admissible parameter values. We refer to $X$ as the reference
coordinate and to $x=\mathcal{G}(X;\boldsymbol{\tau})$ as the associated physical
coordinate. Figure~\ref{fig:dom-map} sketches a representative domain mapping.
% Implicit shock tracking represents $\mathcal{G}$ through an $H^1$ finite element
% discretization bound to a mesh. We instead draw $\mathcal{G}$ from a nonlinear parametric family, which is specified in Section~\ref{sec:optim}.

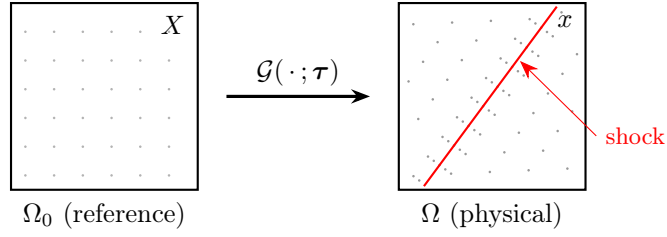
\begin{figure}[t]
  \centering
  \begin{tikzpicture}[>={Stealth[length=2.5mm]},scale=0.95]
    % Reference domain
    \draw[thick] (0,0) rectangle (2.6,2.6);
    \node at (1.3,-0.35) {$\Omega_0$ (reference)};
    \foreach \i in {0.2,0.6,...,2.4} {
      \foreach \j in {0.2,0.6,...,2.4} {
        \fill[gray!60] (\i,\j) circle (0.5pt);
      }
    }
    \node[anchor=north east] at (2.55,2.55) {$X$};
    % Arrow
    \draw[->,very thick] (3.0,1.3) -- node[above]
      {$\mathcal{G}(\,\cdot\,;\boldsymbol{\tau})$} (5.0,1.3); 
    % Physical domain: grid built in the shock-aligned frame. Point layers run
    % parallel to the feature; their normal offsets are tightly spaced for the
    % first two layers on each side and uniform beyond, so the mesh is squeezed
    % onto the shock and relaxes to the reference spacing away from it.
    \draw[thick] (5.4,0) rectangle (8.0,2.6);
    \node at (6.7,-0.35) {$\Omega$ (physical)};
    \pgfmathsetmacro{\dirx}{1.85}  % shock direction (0.35,0.05)--(2.2,2.55)
    \pgfmathsetmacro{\diry}{2.50}
    \pgfmathsetmacro{\dnorm}{sqrt(\dirx*\dirx+\diry*\diry)}
    \pgfmathsetmacro{\tx}{\dirx/\dnorm}   % unit tangent
    \pgfmathsetmacro{\ty}{\diry/\dnorm}
    \pgfmathsetmacro{\nx}{\ty}            % unit normal
    \pgfmathsetmacro{\ny}{-\tx}
    \pgfmathsetmacro{\cx}{1.275}   % shock midpoint
    \pgfmathsetmacro{\cy}{1.30}
    % normal offsets: +-0.10,+-0.20 (squeezed), then uniform 0.42 steps
    \foreach \s in {-1.46,-1.04,-0.62,-0.20,-0.10,0,0.10,0.20,0.62,1.04,1.46} {
      \foreach \m in {-1.87,-1.53,...,1.88} {
        \pgfmathsetmacro{\xx}{\cx + \s*\nx + \m*\tx}
        \pgfmathsetmacro{\yy}{\cy + \s*\ny + \m*\ty}
        \pgfmathparse{(\xx>0.08) && (\xx<2.52) && (\yy>0.08) && (\yy<2.52) ? 1 : 0}
        \ifnum\pgfmathresult=1
          \fill[gray!80] ({5.4+\xx},\yy) circle (0.5pt);
        \fi
      }
    }
    \node[anchor=north east] at (7.95,2.55) {$x$};
    \draw[thick,red] (5.75,0.05) -- (7.6,2.55);
    % label placed outside the domain with a thin leader to the feature
    \node[red,anchor=west,font=\small] (shklbl) at (8.15,0.75) {shock};
    \draw[red,thin,->,shorten >=1pt] (shklbl.west) -- (7.06,1.82);
  \end{tikzpicture}
  \caption{Mapping between the reference domain $\Omega_0$ (left) and
  the physical domain $\Omega$ (right). The map
  $\mathcal{G}(\,\cdot\,;\boldsymbol\tau)$ sends $\partial\Omega_0$ onto
  $\partial\Omega$, so the domain geometry is unchanged, while clustering
  the image of a uniform reference grid along an oblique shock.}
  \label{fig:dom-map}
\end{figure}

The pulled-back solution induced by the mapping~\eqref{eqn:dom-map}, denoted by $\bar{U}:\Omega_0\to\mathbb{R}^m$, is defined on the reference domain as
\begin{equation}
    \bar{U}(X) =
    U\bigl(\mathcal{G}(X;\boldsymbol{\tau})\bigr),
    \qquad X\in\Omega_0.
    \label{eqn:refsol}
\end{equation}
Rewriting~\eqref{eqn:genconslaw2} in the reference domain yields
\begin{equation}
    \bar\nabla \cdot \bar{F}_I(\bar{U};G)
    - \bar\nabla \cdot \bar{F}_V(\bar{U},\bar\nabla\bar{U};G)
    = \bar{S}(\bar{U},\bar\nabla\bar{U};G),
    \qquad X\in\Omega_0,
    \label{eqn:claw-ref}
\end{equation}
with $\bar\nabla = (\partial_{X_1},\dots,\partial_{X_d})$ denoting
the reference-domain gradient. The gradient of the deformation map is
\begin{equation}
    G(X;\boldsymbol{\tau}) =
    \bar\nabla \mathcal{G}(X;\boldsymbol\tau)
    \in \mathbb{R}^{d\times d},
\end{equation}
and the reference-domain fluxes and source term are defined as
\begin{equation}
\begin{aligned}
    \bar{F}_I &: (\bar{W};G) \mapsto \det(G)\, F_I(\bar{W})\, G^{-\top}, \\[2pt]
    \bar{F}_V &: (\bar{W},\bar{Q};G) \mapsto \det(G)\, F_V\!\bigl(\bar{W},\,\bar{Q}\,G^{-1}\bigr)\,G^{-\top}, \\[2pt]
    \bar{S} &: (\bar{W},\bar{Q};G) \mapsto \det(G)\, S\!\bigl(\bar{W},\,\bar{Q}\,G^{-1}\bigr).
\end{aligned}
\label{eqn:refflux}
\end{equation}
Likewise, the physical and reference-domain gradients are related by
\begin{equation}
    \nabla U (x)\bigr|_{x=\mathcal{G}(X;\boldsymbol\tau)}
    = \bar\nabla \bar{U}(X)\, G(X;\boldsymbol{\tau})^{-1}.
    \label{eqn:gradpullback}
\end{equation}
Analogously, the residual operator admits the following
reference-domain counterpart,
\begin{equation}
    \bar{\mathcal{R}}\bigl[\bar W\bigr](X)
    \coloneqq
    \bar\nabla \cdot \bar F_I(\bar W;G)
    - \bar\nabla \cdot \bar F_V(\bar W,\bar\nabla\bar W;G)
    - \bar S(\bar W,\bar\nabla\bar W;G),
    \label{eqn:refresidual}
\end{equation}
so that the mapped conservation law~\eqref{eqn:claw-ref} is equivalent to
$\bar{\mathcal{R}}[\bar U]\equiv 0$ on $\Omega_0$.

\section{Feature-tracking PINNs via nonlinear deformation manifolds}
\label{sec:optim}
% !TEX root = main.tex
In this section, we present the feature-tracking PINN built on the
reference-domain formulation of Section~\ref{subsec:goveq:refdom}. We first
define the nonlinear manifold of deformations from which
$\mathcal{G}(\,\cdot\,;\boldsymbol{\tau})$ is drawn and specify the
parameterization used in this work. We then describe the collocation
distribution induced on the physical domain, the fold barrier that penalizes
folding of the deformation, the mechanism for constraining the deformation map
such that the domain boundaries are preserved, and the resulting training objective.

\subsection{Nonlinear manifolds of deformations}
\label{subsec:optim:manifold}

Let the reference domain be identical to the physical domain, i.e.,
$\Omega_0 = \Omega$, so that the deformation map $\mathcal{G}$ is a self-map that redistributes
points without altering the geometry, and the identity map is an admissible
member of the family. We write $\mathcal{G}$ as the identity plus a parameterized displacement
field,
\begin{equation}
    \mathcal{G}(X;\boldsymbol{\tau})
    = X + \mathbf{d}(X;\boldsymbol{\tau}),
    \qquad
    \mathbf{d}(\,\cdot\,;\boldsymbol{\tau}):\Omega_0\to\mathbb{R}^d,
    \label{eqn:disp-form}
\end{equation}
and refer to the resulting family
\begin{equation}
    \mathcal{M} \coloneqq
    \bigl\{\mathcal{G}(\,\cdot\,;\boldsymbol{\tau}) :
    \boldsymbol{\tau}\in\mathcal{T}\bigr\}
    \label{eqn:manifold}
\end{equation}
as a nonlinear manifold of deformations. Two requirements shape the choice of
$\mathbf{d}$. The displacement must be smooth enough that the second derivatives
demanded by the viscous flux exist, and it must be non-separable, so that the
deformation gradient $G=\bar\nabla\mathcal{G}$ is not restricted to diagonal form
and the reference grid can undergo arbitrary geometric transformations---rotation and shear, for example---rather than merely stretching along the
coordinate axes. The second requirement is what separates $\mathbf{d}$ from the separable
analytical form introduced in our earlier boundary-layer work~\cite{THAKUR2026110419},
where $\mathbf{d}$ is a fixed closed-form expression in a handful of scalar
parameters. Here $\boldsymbol{\tau}$ is far larger and enters nonlinearly, but is
trained against the mapped residual alongside the solution-network weights
exactly as before.

\subsection{Radial basis function-based parameterization}
\label{subsec:optim:rbf}

Throughout this work we represent the displacement by a sum of $N_k$ isotropic
Gaussian radial basis functions,
\begin{equation}
    \mathbf{d}(X;\boldsymbol{\tau})
    = \sum_{k=1}^{N_k} \mathbf{w}_k \,
      \exp\!\left(-\frac{\lVert X - \mathbf{c}_k\rVert_2^2}
                        {2\,\sigma_k^{2}}\right),
    \label{eqn:rbf}
\end{equation}
where $\mathbf{c}_k\in\Omega_0$ is the center of the $k$th kernel,
$\sigma_k>0$ its bandwidth, and $\mathbf{w}_k\in\mathbb{R}^d$ a vector weight
that sets both the magnitude and the direction of the displacement it
induces. The trainable parameter vector collects all three groups,
\begin{equation}
    \boldsymbol{\tau} =
    \bigl\{\mathbf{w}_k,\ \mathbf{c}_k,\ \log\sigma_k\bigr\}_{k=1}^{N_k}
    \in\mathbb{R}^{N_{\boldsymbol\tau}},
    \qquad
    N_{\boldsymbol\tau} = N_k\,(2d+1),
    \label{eqn:tau}
\end{equation}
where the bandwidths are parameterized logarithmically so that $\sigma_k>0$ for
every value of the unconstrained optimization variable. The centers are
initialized on a uniform tensor-product grid over $\Omega_0$ and the bandwidths
at the center spacing, but both are free to move during training.

Three properties of this parameterization deserve emphasis. First, the weights
are initialized to zero, $\mathbf{w}_k^{(0)} = \mathbf{0}$, so that
$\mathcal{G} = \mathrm{id}$, $G = I$, and $\det G = 1$ at initialization: the
FT-PINN begins as a vanilla PINN on a uniform grid and deforms away from it only
as the residual demands. Second, and in contrast to the tensor-product
analytical map, the displacement is not separable. The deformation gradient is non-diagonal and varies across
the domain, so resolution can be aligned with features of arbitrary geometry ---
oblique and curved shocks, for example --- rather than only with the coordinate
axes. Third, the representation is deliberately lightweight.
$N_{\boldsymbol\tau}$ scales linearly in $N_k$ with a small constant.
The price of a learned manifold, whatever its parameterization, is that none of
the structural properties hold automatically, i.e., the map may fold, and it may push
points out of the domain. The next two subsections address both.
\begin{remark}
$G$ is evaluated by automatic differentiation through $\mathcal{G}$, and
$\det G$ and $G^{-1}$ follow by standard linear algebra; these quantities are
required inside the pulled-back fluxes~\eqref{eqn:refflux}.

\end{remark}

\begin{remark}[Other admissible parameterizations]
\label{rem:other-params}
Nothing in the formulation is tied to the employed Gaussian radial-basis representation. Within
the RBF family, any smooth radial function $\phi(r)$---for instance the inverse
multiquadric---may be substituted into~\eqref{eqn:rbf} unchanged. More
substantially, the displacement may equally be taken to be a
fully-connected neural network (FCNN),
\begin{equation}
    \mathbf{d}(X;\boldsymbol{\tau}) = \mathrm{NN}_{\boldsymbol{\tau}}(X),
    \label{eqn:mlp-disp}
\end{equation}
with a suitable activation. This is a strictly more
expressive family, and the entire construction that follows---the fold
barrier~\eqref{eqn:foldbarrier}, the boundary-preserving
projector~\eqref{eqn:tangproj}, the induced collocation
distribution~\eqref{eqn:induced-colloc}, and the joint
objective~\eqref{eqn:refloss}---remains unchanged.
\end{remark}

\subsection{Induced collocation distribution}
\label{subsec:optim:colloc}

Let $\mathcal{X}_c^0 = \{X^{(i)}\}_{i=1}^{N_c}\subset\Omega_0$ denote a fixed set of collocation points defined on 
the reference domain, which, throughout this work, is chosen as a uniform tensor-product grid. 
The corresponding set of collocation points in the physical domain is then given by
\begin{equation}
    \mathcal{X}_c(\boldsymbol{\tau})
    \coloneqq \mathcal{G}(\mathcal{X}_c^0;\boldsymbol{\tau})
    = \bigl\{\mathcal{G}(X^{(i)};\boldsymbol{\tau})\bigr\}_{i=1}^{N_c}.
    \label{eqn:induced-colloc}
\end{equation}
Note that the cardinality of $\mathcal{X}_c(\boldsymbol{\tau})$ is the same as that of
$\mathcal{X}_c^0$ because $\mathcal{G}$ is injective and the reference grid $\mathcal{X}_c^0$ is fixed.
Because the physical collocation points are generated by mapping a fixed
reference-domain grid through $\mathcal{G}(X;\boldsymbol{\tau})$, optimizing
$\boldsymbol{\tau}$ with respect to the reference-domain residual automatically
induces an adaptive distribution of collocation points, without using residual-based
adaptive sampling~\cite{wu2023comprehensive}, NTK-driven batch-adaptive
sampling~\cite{lau2023pinnacle}, or gradient-ascent updates of individual points
based on their pointwise squared residuals~\cite{pacmann2024}. The per-iteration
cost of the forward and backward passes scales linearly with $N_c$, as for a
vanilla PINN. Relative to a vanilla PINN, the additional cost is the evaluation
of $\mathcal{G}$, $G$, $\det G$, and $G^{-1}$ at the fixed reference points and
the gradient of the loss with respect to $\boldsymbol{\tau}$.
In contrast to the aforementioned adaptive-sampling strategies, the proposed 
tracking mechanism is parameterized by the low-dimensional vector $\boldsymbol{\tau}$, and does not require 
independent tracking of individual collocation points or batch-wise point selection.

\subsection{Keeping the map a diffeomorphism: the fold barrier}
\label{subsec:optim:fold}

The reference-domain formulation of Section~\ref{subsec:goveq:refdom} is
predicated on $\mathcal{G}(\,\cdot\,;\boldsymbol\tau)$ being a diffeomorphism,
which requires $\det G > 0$ on $\Omega_0$. This holds at initialization ($\det G = 1$) and for sufficiently small weights,
but nothing in the training objective keeps it true. As the centers migrate to
cluster resolution, $\det G$ can be driven to zero and through it, at which point
the map folds. It then ceases to be invertible, $G^{-1}$
in~\eqref{eqn:refflux} blows up, and the correspondence between the physical and
reference problems is lost.

We prevent this by using a one-sided barrier evaluated on the reference collocation
grid,
\begin{equation}
    \mathcal{L}_{\det}(\boldsymbol{\tau})
    = \frac{1}{N_c}\sum_{i=1}^{N_c}
      \Bigl[
        \max\bigl(0,\ \epsilon_{\det} - \det G(X^{(i)};\boldsymbol{\tau})\bigr)
      \Bigr]^2,
    \qquad \epsilon_{\det} = 10^{-3},
    \label{eqn:foldbarrier}
\end{equation}
which is added to the training objective with a penalty parameter $w_{\det}$.
Because it is one-sided, it is exactly zero
whenever $\det G \ge \epsilon_{\det}$, and it activates only as $\det G$
approaches zero from above, where it pushes the Jacobian determinant back above
the threshold and thereby prevents the sign change that constitutes folding. It
says nothing about $\det G$ being merely small and positive, and so it
does not oppose the desired behavior of the deformation map to sharply contract the volume around
the feature, which is precisely how a uniform set of reference points can be packed into a small region encompassing the feature. 
The barrier costs essentially nothing, since $\det G$ is already
assembled for the mapped fluxes~\eqref{eqn:refflux} and~\eqref{eqn:foldbarrier}
reuses it.

\subsection{Boundary-preservation constraint}
\label{subsec:optim:bdy}

The deformation map must not alter the domain, or mathematically,
$\mathcal{G}(\partial\Omega_0) = \partial\Omega_0$.  On a boundary face with outward unit normal $\mathbf{n}$, this amounts to requiring
that the displacement have no normal component,
$\mathbf{n}\cdot\mathbf{d}=0$. This can be enforced by applying the tangential
projector to the raw RBF displacement,
\begin{equation}
    \mathbf{d} \;\longmapsto\; P\,\mathbf{d},
    \qquad
    P = I - \mathbf{n}\mathbf{n}^{\!\top},
    \label{eqn:tangproj}
\end{equation}
which is the orthogonal projection onto the tangent plane ($P\mathbf{n}=0$), i.e., it removes the normal part and leaves the tangential
displacement untouched. The projection~\eqref{eqn:tangproj} is applied to $\mathbf{d}$ as a mask, so the constraint holds
exactly for every $\boldsymbol{\tau}$ and requires no penalty parameter to tune.

For the axis-aligned box $\Omega_0 = \prod_{\ell=1}^{d}[a_\ell,b_\ell]$ used
throughout this work, every face normal is $\pm\mathbf{e}_\ell$, so
$P = I - \mathbf{e}_\ell\mathbf{e}_\ell^{\!\top}$ acts only on the $\ell$-th
component and the projector reduces to a diagonal mask $M = \mathrm{diag}(m_1,\dots,m_d)$, which leads to the deformation map taking the form,
\begin{equation}
    \mathcal{G}(X;\boldsymbol\tau) = X + M(X)\,\mathbf{d}(X;\boldsymbol\tau).
    \label{eqn:boxmask}
\end{equation} 
To preserve the boundary of $\Omega_0$, it suffices for each mask component $m_\ell$ to vanish on the two faces orthogonal to the $\ell$-th coordinate direction, namely, at $X_\ell=a_\ell$ and $X_\ell=b_\ell$. We therefore define
\begin{equation}
m_\ell(X)
=
\frac{(X_\ell-a_\ell)(b_\ell-X_\ell)}
{\left(\frac{b_\ell-a_\ell}{2}\right)^2}
\in [0,1],
\label{eqn:mask}
\end{equation}
which satisfies $m_\ell(X)=0$ on both boundary faces and attains its maximum value of one at the midpoint of the interval $[a_\ell,b_\ell]$.

\subsection{Training objective on the reference domain}
\label{subsec:optim:loss}
With $\mathcal{X}_b^0=\{X_b^{(j)}\}_{j=1}^{N_b}\subset\partial\Omega_0$ a fixed set
of reference boundary points, the feature-tracking PINN loss is
\begin{equation}
    \bar{\mathcal{L}}(\theta,\boldsymbol{\tau})
    = \frac{w_\mathrm{pde}}{N_c}\sum_{i=1}^{N_c}
      \bigl\Vert
      \bar{\mathcal{R}}\bigl[\hat{\bar U}_\theta\bigr](X^{(i)})
      \bigr\Vert_2^2
    + \frac{w_\mathrm{bc}}{N_b}\sum_{j=1}^{N_b}
      \bigl\Vert
      \bar{\mathcal{B}}\bigl[\hat{\bar U}_\theta\bigr](X_b^{(j)}) - g\bigl(\mathcal{G}(X_b^{(j)};\boldsymbol\tau)\bigr)
      \bigr\Vert_2^2
     + w_{\det}\,\mathcal{L}_{\det}(\boldsymbol{\tau}),
    \label{eqn:refloss}
\end{equation}
where $\hat{\bar U}_\theta:\Omega_0\to\mathbb{R}^m$ is the FCNN evaluated on
reference coordinates and $\mathcal{L}_{\det}$ is the fold
barrier~\eqref{eqn:foldbarrier}. The operator $\bar{\mathcal{B}}$ is the
reference-domain counterpart of the boundary operator $\mathcal{B}$
in~\eqref{eqn:pinn_loss}: for Dirichlet conditions it reduces to evaluation of
$\hat{\bar U}_\theta$ at the reference boundary point, while for conditions
involving derivatives, the physical gradient is recovered from the reference
gradient through~\eqref{eqn:gradpullback}. Note that the boundary data are evaluated at
the physical image $\mathcal{G}(X_b^{(j)};\boldsymbol\tau)$ of each
reference boundary point, which is what allows boundary points to slide while
still being prescribed the correct data. As an example, for the two-state inlet profile of
Section~\ref{subsec:numexp:rr}, a point that slides across the
shock foot correctly switches from the pre- to the post-shock state. The joint
training problem is
\begin{equation}
    (\theta_*,\boldsymbol{\tau}_*)
    \in \argmin_{\theta\in\mathbb{R}^{N_\theta},\,
                  \boldsymbol\tau\in\mathbb{R}^{N_{\boldsymbol\tau}}}
    \bar{\mathcal{L}}(\theta,\boldsymbol{\tau}).
    \label{eqn:jointopt}
\end{equation}
Since $\mathcal{X}_c^0$ remains fixed throughout training, 
no explicit resampling of collocation points is required. Nevertheless, 
the corresponding physical-domain collocation set $\mathcal{X}_c(\boldsymbol{\tau})$ in~\eqref{eqn:induced-colloc} 
undergoes dynamic evolution with changes in $\boldsymbol{\tau}$, allowing the collocation points to 
adaptively track the relevant solution features.

\begin{remark}[Reduction to vanilla PINN]
\label{rem:vanilla-reduction}
When $\boldsymbol{\tau}$ is held fixed, \eqref{eqn:jointopt} reduces to a
standard PINN training problem~\eqref{eqn:pinn_opt} posed on a nonlinearly
deformed collocation set. In particular, holding
$\boldsymbol{\tau}$ at its zero-weight initialization recovers the identity map
and hence exactly the vanilla PINN baseline used for comparison throughout
Section~\ref{sec:num_exp}.
\end{remark}

\section{Joint training of network weights and deformation parameters}
\label{sec:solver}
% !TEX root = main.tex
The non-convex optimization problem defined in Eq.~\eqref{eqn:jointopt} is characterized by the interplay 
between the high-dimensional network parameters $\theta$ and the
deformation parameters $\boldsymbol{\tau}$. We solve this joint optimization problem using Adam~\cite{kingma2014adam} within a mini-batch training framework, treating $\theta$ and $\boldsymbol{\tau}$ as distinct 
parameter groups with  independently scheduled learning rates, denoted by $\eta_\theta$ and $\eta_{\boldsymbol{\tau}}$, respectively.
The reference-domain loss in Eq.~\eqref{eqn:refloss} is computed by making use of automatic
differentiation for the computation of the reference-domain fluxes and source terms, and the gradients with respect to both $\theta$ and
$\boldsymbol{\tau}$ are obtained by a single reverse-mode pass. Since the residual of an untrained network carries no
localized structure for the map to track, we freeze $\boldsymbol\tau$ at the
identity for the first $N_{\mathrm{warm}}$ of the $N_{\mathrm{Adam}}$ total epochs and train $\theta$ alone,
then unfreeze it and solve the joint problem~\eqref{eqn:jointopt} for the
remainder. The Adam moment buffers, the step counter, and the cosine
learning-rate schedule are re-initialized when $\boldsymbol\tau$ is unfrozen,
so that the warm-up and joint phases each run their own cosine decay over
$N_{\mathrm{warm}}$ and $N_{\mathrm{Adam}}-N_{\mathrm{warm}}$ epochs,
respectively. This training
protocol is used consistently across all experiments in
Section~\ref{sec:num_exp}.\par

The parameter count of the deformation field, $N_{\boldsymbol\tau} = N_k(2d+1)$
from~\eqref{eqn:tau}, is small for the two-dimensional problems considered here
with $N_k \le 100$ kernels, $N_{\boldsymbol\tau}\le 500$, against
$N_\theta\sim10^5$ for the solution network. As for a vanilla PINN, the per-iteration cost scales with the number of
collocation points in a batch, $N_B$. The additional overhead of FT-PINN, which scales with $N_k N_B$, comes
from evaluating the RBF sum~\eqref{eqn:rbf} and its derivatives at the batch
points and from assembling $\det G$ and $G^{-1}$
at each point. Both are modest relative to the
FCNN forward and backward passes, and the fold barrier reuses the already
assembled $\det G$.

\begin{algorithm}[H]
\caption{Joint training of the feature-tracking PINN}
\label{alg:trainer}
\begin{algorithmic}[1]
\REQUIRE{Reference grid $\mathcal{X}_c^0\subset\Omega_0$, boundary points
$\mathcal{X}_b^0\subset\partial\Omega_0$, initial $\theta^{(0)}$ and
$\mathbf{w}_k^{(0)}=\mathbf{0}$ (identity map), horizons
$N_{\mathrm{warm}}, N_{\mathrm{Adam}}$, rates
$\eta_\theta,\eta_{\boldsymbol\tau}$, batch size $N_B$, weights
$w_\mathrm{pde}, w_\mathrm{bc}, w_{\det}$.}
\ENSURE{Trained $(\theta_*,\boldsymbol{\tau}_*)$.}
\STATE $(\theta_*,\boldsymbol{\tau}_*) \leftarrow
        (\theta^{(0)},\boldsymbol{\tau}^{(0)})$.
\FOR{$n = 1,\dots,N_{\mathrm{Adam}}$}
  \STATE \textbf{if} $n \le N_{\mathrm{warm}}$ \textbf{then} freeze
         $\boldsymbol{\tau}$ (train as vanilla PINN) \textbf{else} unfreeze
         $\boldsymbol{\tau}$.
  \STATE \textbf{if} $n \in \{1,\,N_{\mathrm{warm}}+1\}$ \textbf{then} set the
         step counter $k \leftarrow 0$, initialize the moment buffers
         $(m_p,v_p) \leftarrow (0,0)$ for each active parameter group $p$, and
         restart the cosine learning-rate schedule.
  \STATE Shuffle $\mathcal{X}_c^0$ into mini-batches of size $N_B$.
  \FORALL{mini-batches $\mathcal{N}$}
    \STATE Evaluate $\bar{\mathcal{L}}_\mathcal{N}$,
           $\nabla_{\theta}\bar{\mathcal{L}}_\mathcal{N}$, and
           $\nabla_{\boldsymbol{\tau}}\bar{\mathcal{L}}_\mathcal{N}$
           from~\eqref{eqn:refloss}.
    \STATE $k \leftarrow k+1$.
    \STATE For each active parameter group $p\in\{\theta_*,\boldsymbol{\tau}_*\}$ with gradient $g_p=\nabla_p\bar{\mathcal{L}}_\mathcal{N}$, moment buffers $(m_p,v_p)$, and rate $\eta_p\in\{\eta_\theta,\eta_{\boldsymbol\tau}\}$, apply the Adam~\cite{kingma2014adam} update
    \[
      m_p \leftarrow \beta_1 m_p + (1-\beta_1)\,g_p, \quad
      v_p \leftarrow \beta_2 v_p + (1-\beta_2)\,g_p\odot g_p, \quad
      p \leftarrow p - \eta_p\,\frac{m_p/(1-\beta_1^{k})}{\sqrt{v_p/(1-\beta_2^{k})}+\varepsilon},
    \]
    with decay rates $(\beta_1,\beta_2)=(0.9,0.999)$ and $\varepsilon=10^{-8}$. The buffers $(m_p,v_p)$ have the same dimension as $p$ and are maintained separately for each parameter group.
  \ENDFOR
  \STATE Update $\eta_\theta,\eta_{\boldsymbol\tau}$ by one step of the cosine decay learning-rate schedule.
\ENDFOR
\RETURN $(\theta_*,\boldsymbol{\tau}_*)$.
\end{algorithmic}
\end{algorithm}

\section{Numerical experiments}
\label{sec:num_exp}
% !TEX root = main.tex
In this section we benchmark the proposed feature-tracking PINN framework on
four problems whose solutions contain shocks. In each experiment we compare the
following two configurations:
\begin{description}[leftmargin=1.25em,itemsep=1pt]
  \item[FT-PINN:] the proposed feature-tracking PINN, with a deformation map based on a mixture of Gaussian RBFs trained jointly with the FCNN using Algorithm~\ref{alg:trainer}.
  \item[Vanilla PINN:] a PINN with no deformation map, trained on a uniform
        collocation grid in $\Omega$ with otherwise identical architecture,
        optimizer, collocation budget, and training horizon.
\end{description}

\paragraph{On the role of viscosity and the assessment of accuracy}
All the problems below are solved as a viscous conservation law, where a viscous flux 
with a diffusion coefficient $\nu$ acts on the conserved variables. The role of $\nu$ differs 
depending on the problem. For the problems involving the Burgers equation in Sections~\ref{subsec:numexp:burgst} 
and~\ref{subsec:numexp:burgdec}, it is a physical property of the governing equation itself, 
whereas for the problems based on the Euler equations in Sections~\ref{subsec:numexp:sst} and~\ref{subsec:numexp:rr}, 
the underlying physical model is inviscid with true discontinuities, 
and $\nu$ is introduced purely as an artificial viscosity (regularization) 
so that pointwise derivatives and the strong-form residual remain well-defined. 
In both cases, the solution features a smooth internal layer of width $O(\nu)$ 
instead of a jump, and the exact location of this layer is not known beforehand. Furthermore,
the use of viscosity affects the choice of error measure based on the reference solution. 
For the merging-shock problem of Section~\ref{subsec:numexp:burgst} no analytical solution is known, and the
reference is obtained by solving the same viscous equation with the same $\nu$
using a high-resolution finite volume scheme. The
remaining three admit exact inviscid solutions, which are discontinuous where
the computed profile is a layer of finite thickness, the two agreeing only in
the limit $\nu\to0$. A pointwise error taken against such a reference is
floored by the regularization in two ways. Away from the feature the viscous
solution differs from the inviscid one by $O(\nu)$, so no approximation of the
viscous equation can be closer than that to the reference. Across the feature
the layer of width $O(\nu)$ is measured against a jump, which produces an
$O(1)$ spike of width $O(\nu)$ for any solution of the viscous equation,
however accurate. Such an error therefore reports the regularization rather
than the quality of the approximation once the latter is resolved to the level
$\nu$. Accuracy in Sections~\ref{subsec:numexp:burgdec}, \ref{subsec:numexp:sst}
and~\ref{subsec:numexp:rr} is therefore assessed primarily through solution
profiles across the feature, which expose the jump amplitude, the layer width,
and the feature location directly, together with contours of the full field
against the reference. For the decelerating shock of
Section~\ref{subsec:numexp:burgdec} we additionally report the pointwise error
on a logarithmic scale with the level $\nu$ marked as a horizontal line, so
that the floor set by the regularization is visible alongside the error itself.\par

For each experiment we also show the uniform reference collocation grid
$\mathcal{X}_c^0$ alongside its image
$\mathcal{X}_c(\boldsymbol{\tau}_*)=\mathcal{G}(\mathcal{X}_c^0;\boldsymbol{\tau}_*)$
under the trained map to illustrate how the deformation concentrates points on the feature.

\paragraph{Common settings}
Unless stated otherwise, the solution network is a 5-layer FCNN of width 128
with GELU activations and a random Fourier feature
embedding~\cite{tancik2020fourier} of the input coordinates with
$n_{\mathrm{Fourier}}$ random Fourier features whose frequencies are drawn
from a zero-mean Gaussian of standard deviation $\sigma_{\mathrm{Fourier}}$,
the displacement field is a mixture of Gaussian RBFs~\eqref{eqn:rbf} with
zero-initialized weights and initial centers on a uniform grid, and training follows
Algorithm~\ref{alg:trainer} with Adam, $\eta_\theta=\eta_{\boldsymbol\tau}=10^{-3}$,
and cosine decay learning rate scheduler. Problem-specific settings are collected in
Table~\ref{tab:hyper}.

\begin{table}[h]
  \centering
  \caption{Problem and training parameters for the four numerical experiments.
  Grid size is the number of interior tensor-product stations per direction,
  so $N_c$, the number of interior collocation points (faces excluded), is
  its product; $N_b$ is the
  number of boundary points on which the boundary loss is evaluated, so faces
  that carry no condition (the outflow face $t=T$ of the merging-shock problem,
  all but $t=0$ for the shock tube, and $x_2=0$, $x_1=2$ for the regular
  reflection) contribute to neither. The second vanilla PINN of
  Section~\ref{subsec:numexp:rr} replaces the interior grid by $N_c=12{,}000$
  random points with $N_B=6000$ and is otherwise identical to the last
  column.}
  \label{tab:hyper}
  \footnotesize
  \begin{tabular}{lcccc}
    \toprule
     & Burgers (merging) & Burgers (decel.) & Euler shock tube & Regular reflection \\
     & \S\ref{subsec:numexp:burgst} & \S\ref{subsec:numexp:burgdec}
     & \S\ref{subsec:numexp:sst} & \S\ref{subsec:numexp:rr} \\
    \midrule
    domain $\Omega$        & $[-1,1]\times[0,1]$ & $[-1,1]\times[0,1]$
                           & $[0,1]\times[0,0.2]$ & $[0,2]\times[0,1]$ \\
    $m$ (state size)       & $1$ & $1$ & $3$ & $4$ \\
    $\nu$                  & $10^{-3}$ & $10^{-3}$ & $8\times10^{-4}$ & $7.5\times10^{-4}$ \\
    Grid size              & $60\times50$ & $60\times50$ & $70\times50$ & $98\times58$ \\
    $N_c$                  & $3000$ & $3000$ & $3500$ & $5684$ \\
    $N_b$                  & $166$ & $228$ & $72$ & $160$ \\
    batch size $N_B$       & $1500$ & $1500$ & $1750$ & $3000$ \\
    $N_{\mathrm{Adam}}$    & $10^4$ & $10^4$ & $2\times10^4$ & $2\times10^4$ \\
    $N_{\mathrm{warm}}$    & $1000$ & $1000$ & $2000$ & $2000$ \\
    $N_k$                  & $81$ & $81$ & $81$ & $100$ \\
    $N_{\boldsymbol\tau}$  & $405$ & $405$ & $405$ & $500$ \\
    $n_{\mathrm{Fourier}}$ & $128$ & $128$ & $64$ & $64$ \\
    $\sigma_{\mathrm{Fourier}}$ & $5$ & $5$ & $10$ & $5$ \\
    $w_{\det}$             & $0.8$ & $100$ & $0.1$ & $0.1$ \\
    $w_\mathrm{pde}$       & $1$ & $1$ & $1$ & $1$ \\
    $w_\mathrm{bc}$        & $1$ & $10$ & $10$ & $10$ \\
    \bottomrule
  \end{tabular}
\end{table}

% =====================================================================
\subsection{Space-time viscous Burgers with merging shocks}
\label{subsec:numexp:burgst}

\paragraph{Setup}
Consider the one-dimensional scalar viscous Burgers equation solved in a
space-time domain $\Omega=[-1,1]\times[0,1]$ with coordinates $(x_1,t)$, following Section~\ref{subsec:goveq:spacetime}. 
In the formalism of \eqref{eqn:genconslaw2}, the state, fluxes, and source are
\begin{equation}
    U = q,\quad
    F_I(U) = \begin{bmatrix}\tfrac{1}{2}q^2 & q\end{bmatrix},\qquad
    F_V(U,\nabla U) = \begin{bmatrix}\nu\, \partial_{x_1} q & 0\end{bmatrix},\qquad
    S = 0,
    \label{eqn:burgst-flux}
\end{equation}
with $\nu = 10^{-3}$. The initial condition is the following linear combination of two Gaussian bumps,
\begin{equation}
    q(x_1,0) = q_0(x_1) = 1.2\,e^{-(x_1+0.5)^2/0.025} - e^{-(x_1-0.5)^2/0.025},
    \label{eqn:burgst-ic}
\end{equation}
whose compressive regions steepen into two shocks that propagate toward one
another and eventually merge into a single shock. Homogeneous Dirichlet conditions, $q(\pm1,t)=0$,
are imposed at both spatial boundaries, and the initial condition is prescribed on the face $t=0$.
This problem provides a challenging test of the deformation map, as it must resolve
features that are both moving and topologically changing.

\paragraph{Reference solution}
The reference solution is obtained by time marching the viscous equation with a second order
cell-centered finite volume method (FVM) on a uniform grid of $3000$ cells in $x_1$
($\Delta x_1 = 2/3000 \approx 6.7\times10^{-4}$). The inviscid flux uses MUSCL linear reconstruction
with the monotonized-central slope limiter and an exact Riemann solver at the cell faces, the viscous
flux is discretized centrally, and time integration uses the explicit third order strong
stability-preserving Runge--Kutta scheme \cite{gottlieb2001strong} with a fixed step
$\Delta t \approx 7.9\times10^{-5}$ for a total of $12{,}600$ steps to reach $t=1$. The solution is stored every $\Delta t=0.01$, and
the space-time field used for the contour comparison is obtained by linear interpolation of these
snapshots.

\paragraph{Results}
Figure~\ref{fig:burgst_contours} shows a comparison of solution contours over the space-time domains
for both vanilla and FT-PINN against the finite volume reference. The vanilla PINN
not only smears both fronts over a wide band but fails to capture the post-merge shock,
while FT-PINN reproduces the reference structure, including the merger.
Furthermore, Figure~\ref{fig:burgst_slices} compares solution profiles at $t=0.1,0.3,0.6,0.9$.
At every time instance, the vanilla PINN clips the amplitude of shocks and broadens
them, whereas the FT-PINN profile remains superposed on the reference. Figure~\ref{fig:burgst_points} shows the induced
collocation distribution. Points leave the uniform lattice and gather along both
space-time shock trajectories, with the dense concentration where the shocks are located.

\begin{figure}[h]
  \centering
  \resizebox{0.95\textwidth}{!}{\begin{tikzpicture}
\begin{axis}[
xmin=-1,
xmax=1,
ymin=0,
ymax=1,
xlabel=$x$,
ylabel=$t$,
xtick={-1, 0, 1},
ytick={0, 0.5, 1},
width=0.30\textwidth,
height=4.2cm,
name=burgstc0]
\addplot []
graphics [xmin=-1,xmax=1,ymin=0,ymax=1] { 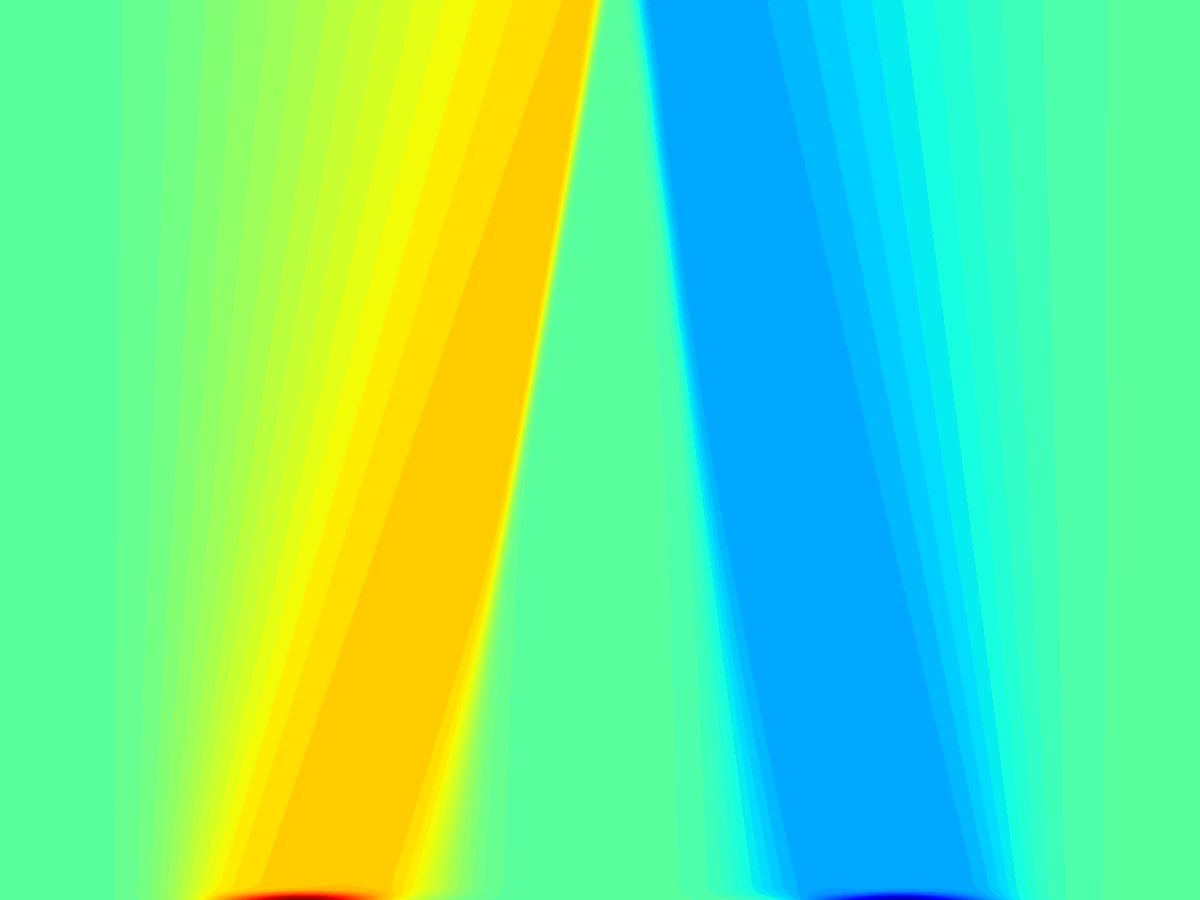};

\end{axis}
\begin{axis}[
xmin=-1,
xmax=1,
ymin=0,
ymax=1,
xlabel=$x$,
ylabel= ,
xtick={-1, 0, 1},
ytick={0, 0.5, 1},
width=0.30\textwidth,
height=4.2cm,
name=burgstc1,
anchor=west,
at={(burgstc0.east)},
xshift=0.85cm]
\addplot []
graphics [xmin=-1,xmax=1,ymin=0,ymax=1] { 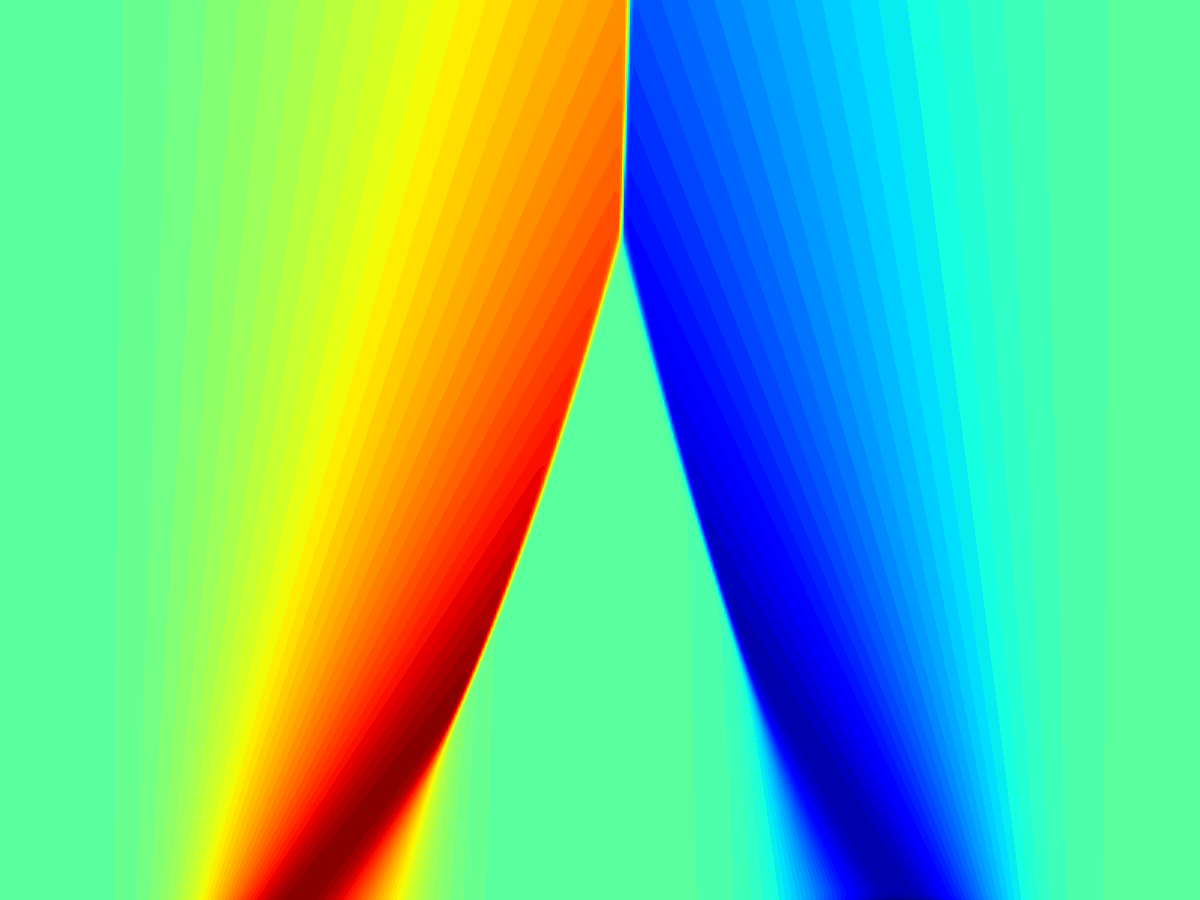};

\end{axis}
\begin{axis}[
xmin=-1,
xmax=1,
ymin=0,
ymax=1,
xlabel=$x$,
ylabel= ,
xtick={-1, 0, 1},
ytick={0, 0.5, 1},
width=0.30\textwidth,
height=4.2cm,
name=burgstc2,
anchor=west,
at={(burgstc1.east)},
xshift=0.85cm]
\addplot []
graphics [xmin=-1,xmax=1,ymin=0,ymax=1] { 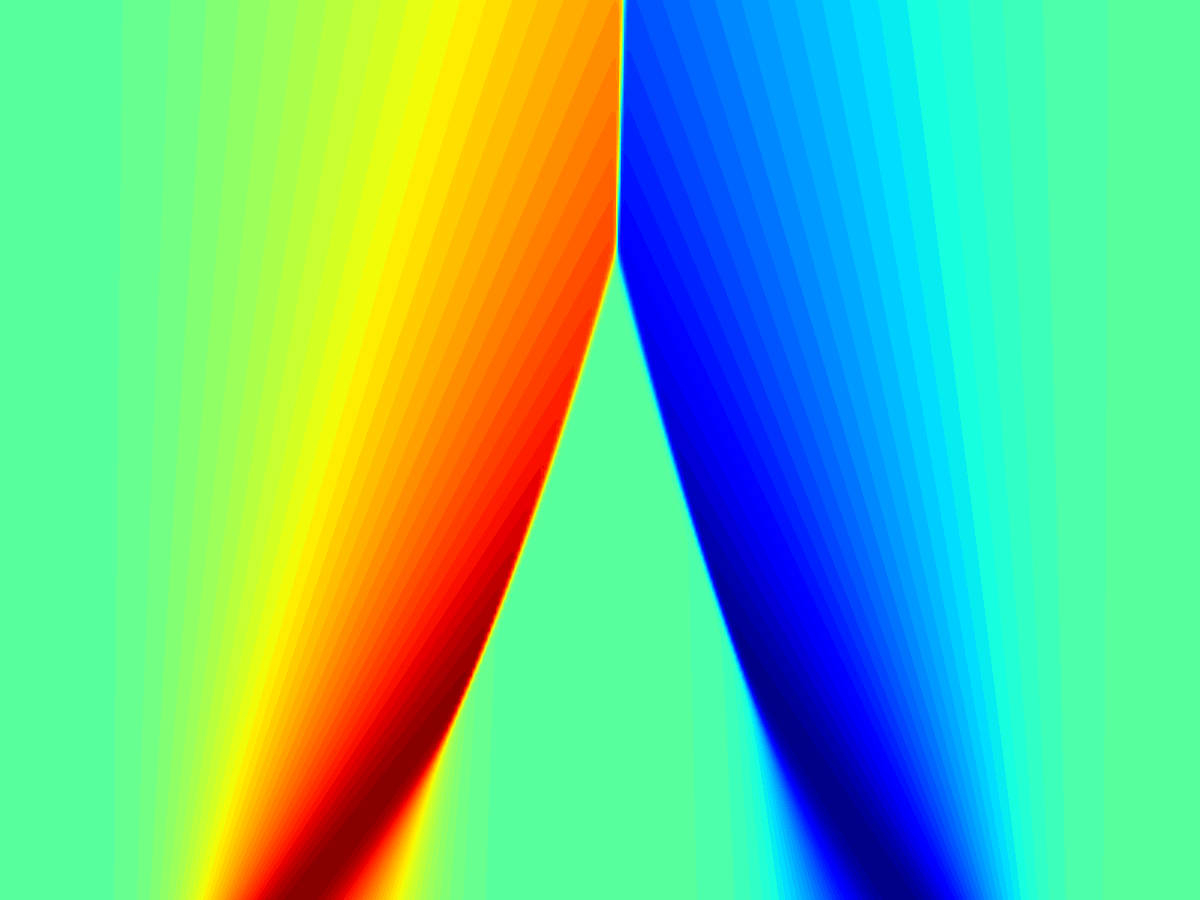};

\end{axis}
\node[anchor=north, font=\small] at ($(burgstc0.south) + (0, -1.05cm)$) {(a)};
\node[anchor=north, font=\small] at ($(burgstc1.south) + (0, -1.05cm)$) {(b)};
\node[anchor=north, font=\small] at ($(burgstc2.south) + (0, -1.05cm)$) {(c)};
\begin{axis}[
hide axis, scale only axis,
xmin=0, xmax=1, ymin=0, ymax=1,
height=0pt, width=0pt,
colormap/jet,
colorbar horizontal,
point meta min=-1, point meta max=1.2,
colorbar style={
  width=0.62\textwidth, height=0.25cm,
  xtick={-1.00, -0.45, 0.10, 0.65, 1.20},
  ticklabel style={font=\small},
  at={(burgstc1.south)}, anchor=north, yshift=-1.75cm
}]
\addplot [draw=none] coordinates {(0,0)};
\end{axis}

\end{tikzpicture}}
  \caption{Space-time viscous Burgers with merging shocks: contours of $q(x_1,t)$
  over $\Omega=[-1,1]\times[0,1]$ for \textbf{(a)} the vanilla PINN,
  \textbf{(b)} FT-PINN, and \textbf{(c)} the high-resolution finite volume
  reference.}
  \label{fig:burgst_contours}
\end{figure}
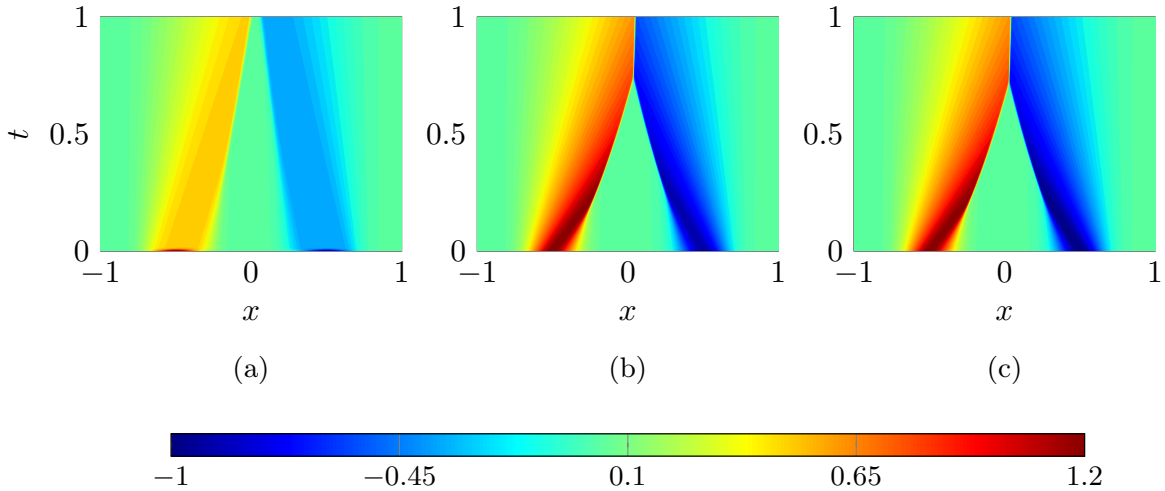

\begin{figure}[h]
  \centering
  \resizebox{0.98\textwidth}{!}{\input{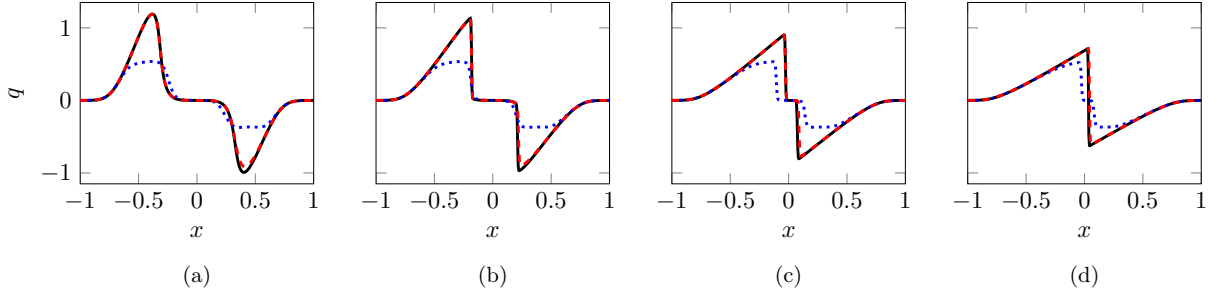}}
  \caption{Space-time viscous Burgers equation with merging shocks: solution
  profiles $q(x_1)$ at \textbf{(a)} $t=0.1$, \textbf{(b)} $t=0.3$,
  \textbf{(c)} $t=0.6$, and \textbf{(d)} $t=0.9$, comparing the finite volume
  reference (\ref{line:burgst_ref}), the vanilla PINN (\ref{line:burgst_A}),
  and FT-PINN (\ref{line:burgst_B}).}
  \label{fig:burgst_slices}
\end{figure}

\begin{figure}[h]
  \centering
  \resizebox{0.9\textwidth}{!}{\begin{tikzpicture}
\begin{axis}[
xmin=-1,
xmax=1,
ymin=0,
ymax=1,
xlabel=$x$,
ylabel=$t$,
xtick={-1, -0.5, 0, 0.5, 1},
ytick={0, 0.25, 0.5, 0.75, 1},
width=0.40\textwidth,
height=4.4cm,
name=burgstp0]
\addplot []
graphics [xmin=-1,xmax=1,ymin=0,ymax=1] { 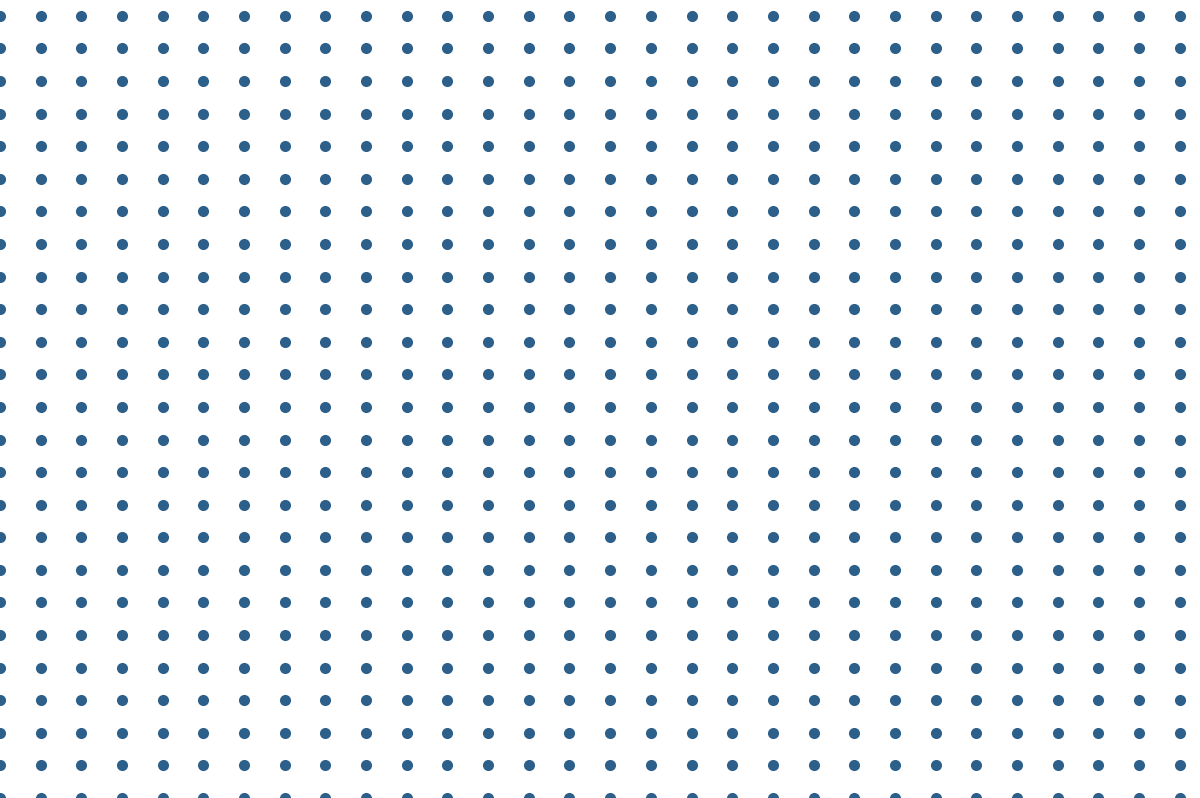};

\end{axis}
\begin{axis}[
xmin=-1,
xmax=1,
ymin=0,
ymax=1,
xlabel=$x$,
ylabel= ,
xtick={-1, -0.5, 0, 0.5, 1},
ytick={0, 0.25, 0.5, 0.75, 1},
width=0.40\textwidth,
height=4.4cm,
name=burgstp1,
anchor=west,
at={(burgstp0.east)},
xshift=1.15cm]
\addplot []
graphics [xmin=-1,xmax=1,ymin=0,ymax=1] { 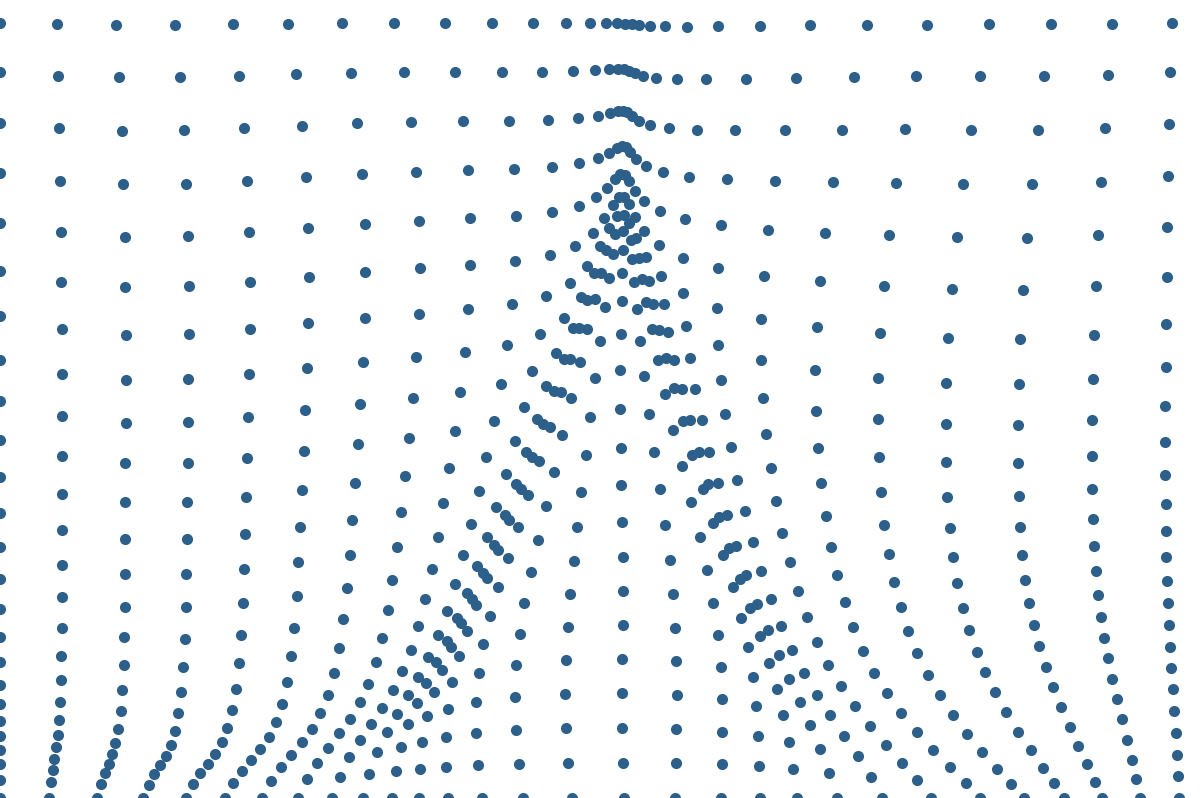};

\end{axis}
\node[anchor=north, font=\small] at ($(burgstp0.south) + (0, -1.05cm)$) {(a)};
\node[anchor=north, font=\small] at ($(burgstp1.south) + (0, -1.05cm)$) {(b)};
\end{tikzpicture}}
  \caption{Space-time viscous Burgers equation with merging shocks:
  \textbf{(a)} uniform reference collocation grid $\mathcal{X}_c^0$ and
  \textbf{(b)} its image $\mathcal{G}(\mathcal{X}_c^0;\boldsymbol{\tau}_*)$
  after joint training.}
  \label{fig:burgst_points}
\end{figure}

% =====================================================================
\subsection{Space-time Burgers with a decelerating shock}
\label{subsec:numexp:burgdec}

\paragraph{Setup}
Next, we consider the viscous Burgers equation ($\nu=10^{-3}$) with a decelerating shock~\cite{shi2022implicit} in the domain $\Omega=[-1,1]\times[0,1]$,
with $x_1$ the spatial and $t$ the temporal coordinate, and the initial condition
\begin{equation}
    q_0(x_1) = 2\,(x_1+1)^2\,\bigl(1 - H_\varepsilon(x_1)\bigr),
    \qquad
    H_\varepsilon(x_1) = \tfrac{1}{2}\bigl(1 + \tanh(x_1/\varepsilon)\bigr),
    \label{eqn:burgdec-ic}
\end{equation}
which is a quadratic ramp rising from $0$ at $x_1=-1$ to $2$ at $x_1=0^-$ followed by a jump
down to zero. The Heaviside jump is smoothed over $\varepsilon=7.5\times10^{-3}$. Dirichlet conditions $q=0$ are imposed at $x_1=\pm1$ (both are
consistent with the exact solution over the time window), the face $t=0$ carries
the initial condition, and a homogeneous
Neumann condition is applied at $t=1$ for stability.

\paragraph{Exact solution}
In the sharp-interface limit $\varepsilon\to0$, the corresponding inviscid
problem admits a closed-form entropy solution, derived in
\ref{app:burgdec}. In terms of $s = x_1+1$, it is given by
\begin{equation}
    q(x_1,t) =
    \begin{cases}
      \dfrac{4s^2}{4st + 1 + \sqrt{1+8st}}, & x_1 < x_s(t),\\[8pt]
      0, & x_1 \ge x_s(t),
    \end{cases}
    \label{eqn:burgdec-exact}
\end{equation}
while the shock traces the parametric curve
\begin{equation}
    t(a) = \frac{1-a^3}{3a^4},
    \qquad
    x_s(a) = (a-1) + \frac{2\,(1-a^3)}{3a^2},
    \qquad a \in (0,1],
    \label{eqn:burgdec-shock}
\end{equation}
which satisfies the Rankine--Hugoniot condition exactly.
\paragraph{Results}
Figure~\ref{fig:burgdec_slices} compares solution profiles at $t=0.05,0.35,0.65,0.95$
against~\eqref{eqn:burgdec-exact}. The FT-PINN profile follows the exact solution closely at every time, capturing both the quadratic
ramp and a sharp jump at the correct location. The vanilla PINN, by contrast,
fails outright, and never develops a
jump at all. Figure~\ref{fig:burgdec_errors} shows the corresponding pointwise
error $|q-q_{\mathrm{exact}}|$ on a logarithmic scale, with the horizontal
line marking $\nu$. Outside the shock layer the FT-PINN error is of the order
of $\nu$ at all four times, which is the level at which the viscous solution
itself departs from the inviscid reference, so the approximation is resolved
down to the floor set by the regularization. The narrow spike at the shock
position is the viscous layer of width $O(\nu)$ measured against a jump, and
is present for any solution of the viscous equation. The vanilla PINN error,
by contrast, is $O(1)$ over a wide band about the shock at every time.
Figure~\ref{fig:burgdec_contours} reinforces the same conclusion through the solution contours over
the whole space-time domain, and Figure~\ref{fig:burgdec_points} shows the collocation
distribution tracking the curved, decelerating shock path.

\begin{figure}[h]
  \centering
  \resizebox{0.98\textwidth}{!}{\input{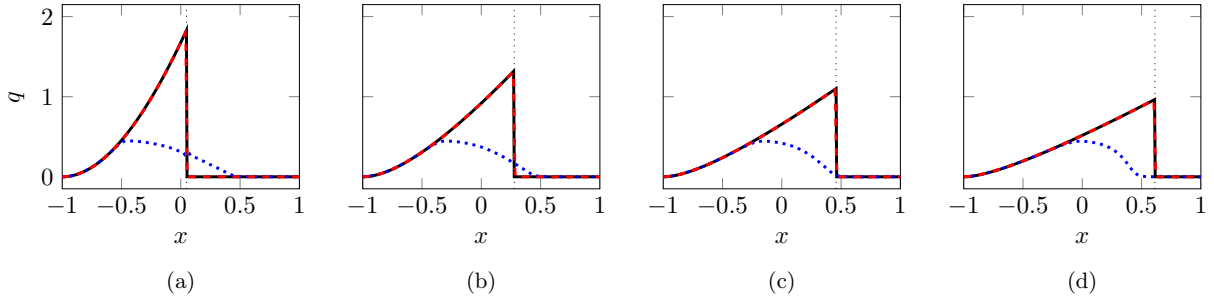}}
  \caption{Decelerating Burgers shock: solution profiles $q(x_1)$ at
  \textbf{(a)} $t=0.05$, \textbf{(b)} $t=0.35$, \textbf{(c)} $t=0.65$, and
  \textbf{(d)} $t=0.95$, comparing the exact inviscid
  solution (\ref{line:burgdec_ref}), the vanilla PINN
  (\ref{line:burgdec_A}), and FT-PINN (\ref{line:burgdec_B}). The dotted
  vertical line in each panel marks the exact shock position at that
  time.}
  \label{fig:burgdec_slices}
\end{figure}

\begin{figure}[h]
  \centering
  \resizebox{0.98\textwidth}{!}{\input{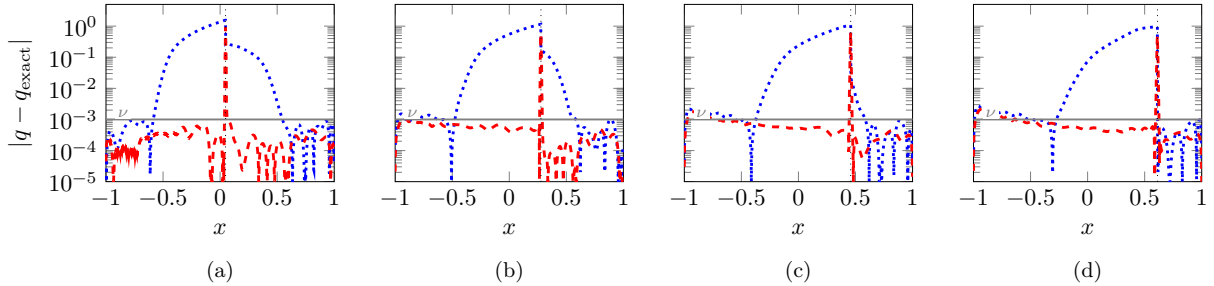}}
  \caption{Decelerating Burgers shock: pointwise error $|q-q_{\mathrm{exact}}|$
  against the exact inviscid solution~\eqref{eqn:burgdec-exact} on a
  logarithmic scale at \textbf{(a)} $t=0.05$, \textbf{(b)} $t=0.35$,
  \textbf{(c)} $t=0.65$, and \textbf{(d)} $t=0.95$, for the vanilla PINN
  (\ref{line:burgdec_err_A}) and FT-PINN (\ref{line:burgdec_err_B}). The
  horizontal line (\ref{line:burgdec_err_nu}) marks the viscosity
  $\nu=10^{-3}$, the floor below which the error measures the regularization
  rather than the approximation, and the dotted vertical line marks the exact
  shock position. Errors below $10^{-5}$ are clipped to the lower axis limit.}
  \label{fig:burgdec_errors}
\end{figure}

\begin{figure}[h]
  \centering
  \resizebox{0.95\textwidth}{!}{\begin{tikzpicture}
\begin{axis}[
xmin=-1,
xmax=1,
ymin=0,
ymax=1,
xlabel=$x$,
ylabel=$t$,
xtick={-1, 0, 1},
ytick={0, 0.5, 1},
width=0.30\textwidth,
height=4.2cm,
name=burgdecc0]
\addplot []
graphics [xmin=-1,xmax=1,ymin=0,ymax=1] { 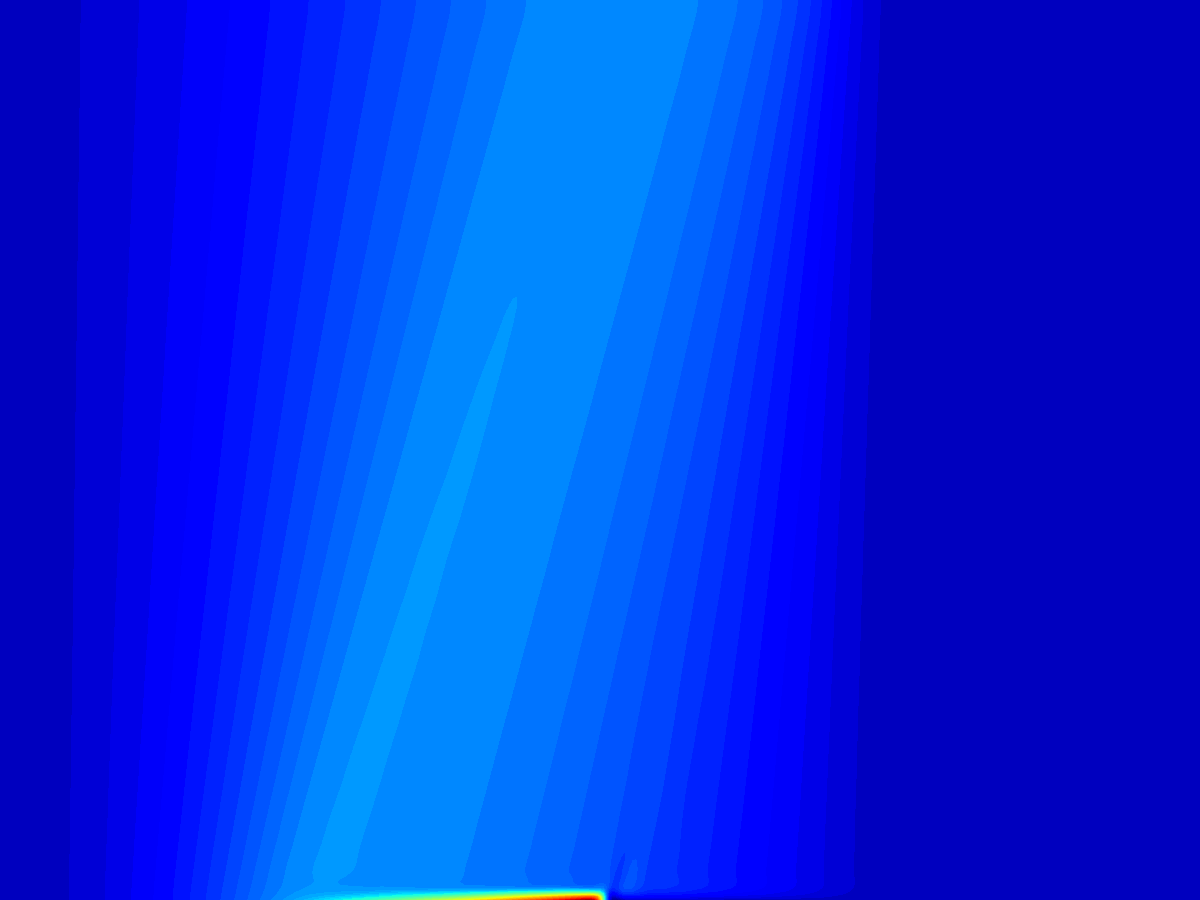};

\end{axis}
\begin{axis}[
xmin=-1,
xmax=1,
ymin=0,
ymax=1,
xlabel=$x$,
ylabel= ,
xtick={-1, 0, 1},
ytick={0, 0.5, 1},
width=0.30\textwidth,
height=4.2cm,
name=burgdecc1,
anchor=west,
at={(burgdecc0.east)},
xshift=0.85cm]
\addplot []
graphics [xmin=-1,xmax=1,ymin=0,ymax=1] { 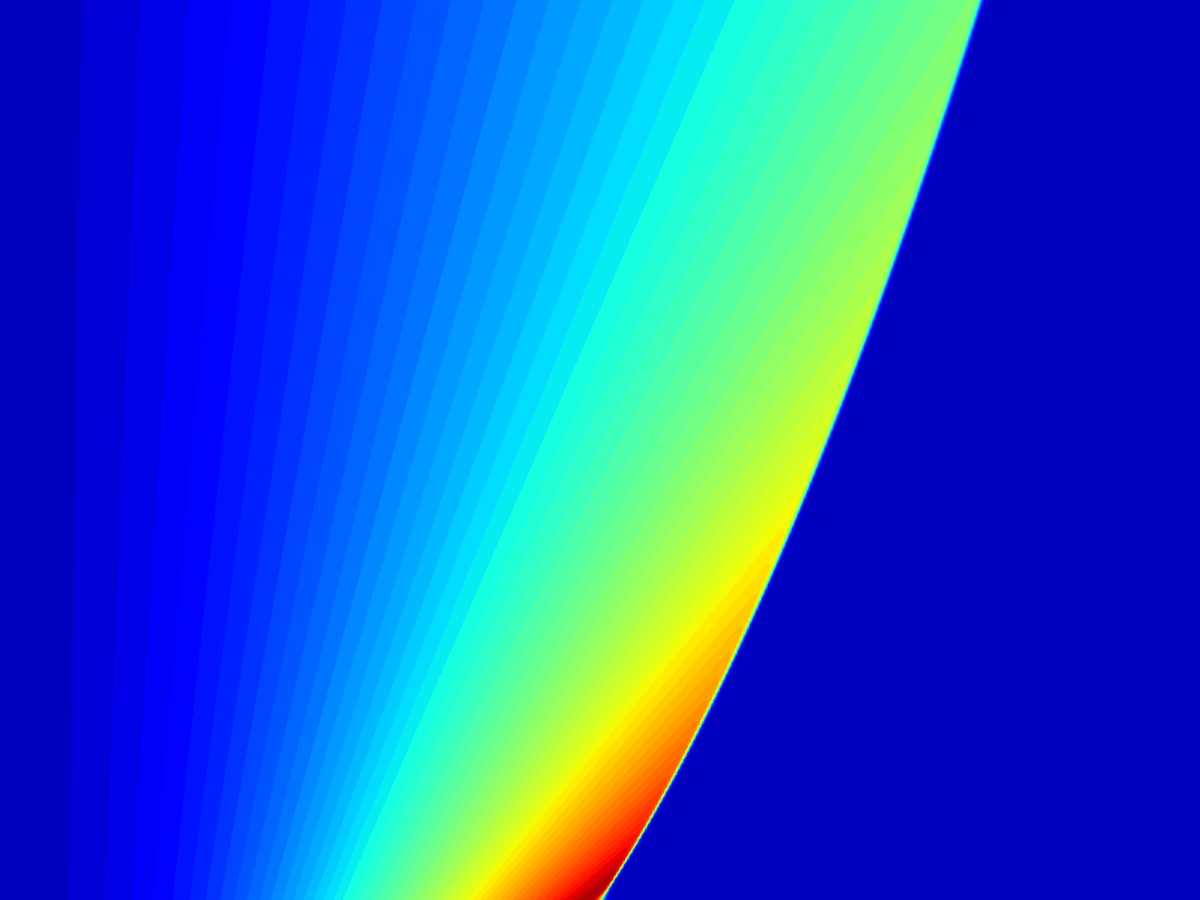};

\end{axis}
\begin{axis}[
xmin=-1,
xmax=1,
ymin=0,
ymax=1,
xlabel=$x$,
ylabel= ,
xtick={-1, 0, 1},
ytick={0, 0.5, 1},
width=0.30\textwidth,
height=4.2cm,
name=burgdecc2,
anchor=west,
at={(burgdecc1.east)},
xshift=0.85cm]
\addplot []
graphics [xmin=-1,xmax=1,ymin=0,ymax=1] { 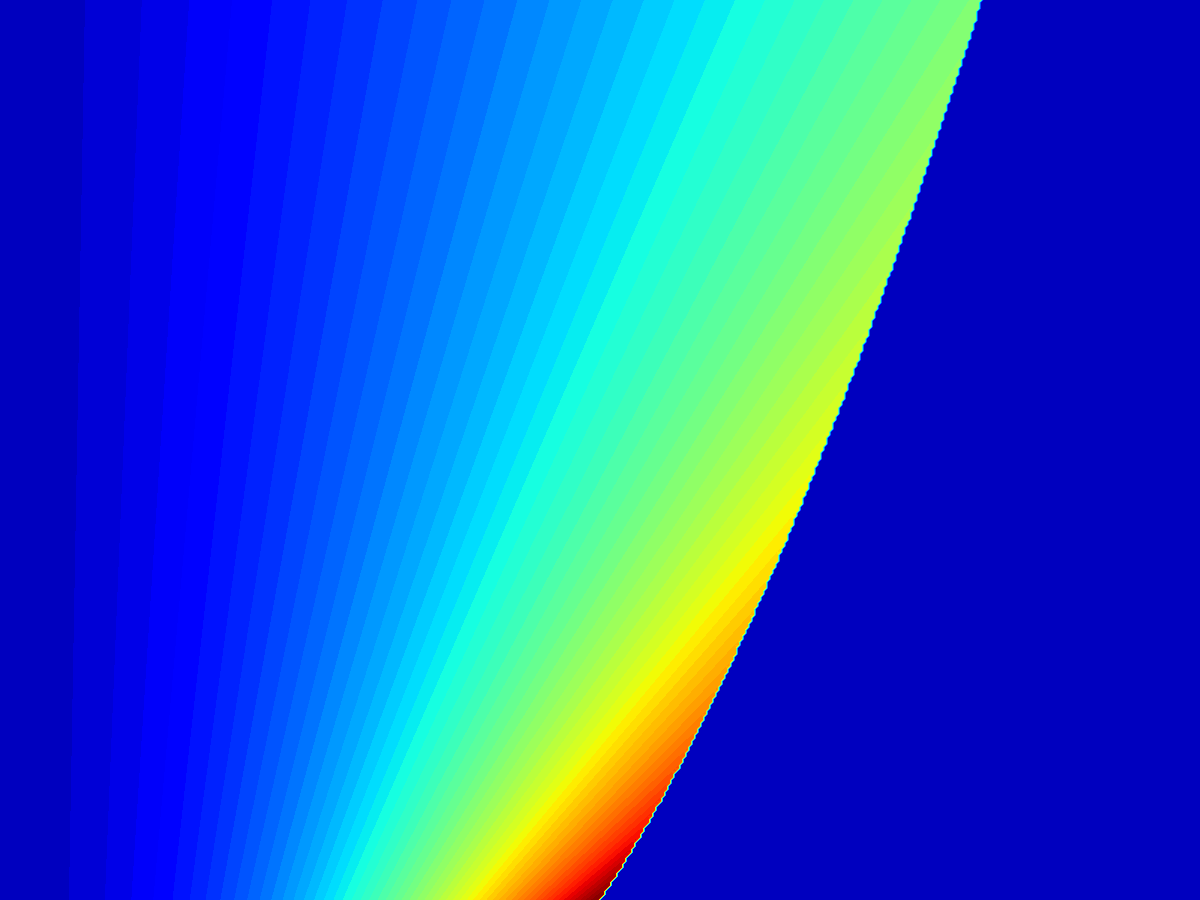};

\end{axis}
\node[anchor=north, font=\small] at ($(burgdecc0.south) + (0, -1.05cm)$) {(a)};
\node[anchor=north, font=\small] at ($(burgdecc1.south) + (0, -1.05cm)$) {(b)};
\node[anchor=north, font=\small] at ($(burgdecc2.south) + (0, -1.05cm)$) {(c)};
\begin{axis}[
hide axis, scale only axis,
xmin=0, xmax=1, ymin=0, ymax=1,
height=0pt, width=0pt,
colormap/jet,
colorbar horizontal,
point meta min=-0.1127, point meta max=1.98,
colorbar style={
  width=0.62\textwidth, height=0.25cm,
  xtick={-0.11, 0.41, 0.93, 1.46, 1.98},
  ticklabel style={font=\small},
  at={(burgdecc1.south)}, anchor=north, yshift=-1.75cm
}]
\addplot [draw=none] coordinates {(0,0)};
\end{axis}

\end{tikzpicture}}
  \caption{Decelerating Burgers shock: contours of $q(x_1,t)$ for \textbf{(a)}
  the vanilla PINN, \textbf{(b)} FT-PINN, and \textbf{(c)} the exact inviscid
  solution~\eqref{eqn:burgdec-exact}.}
  \label{fig:burgdec_contours}
\end{figure}
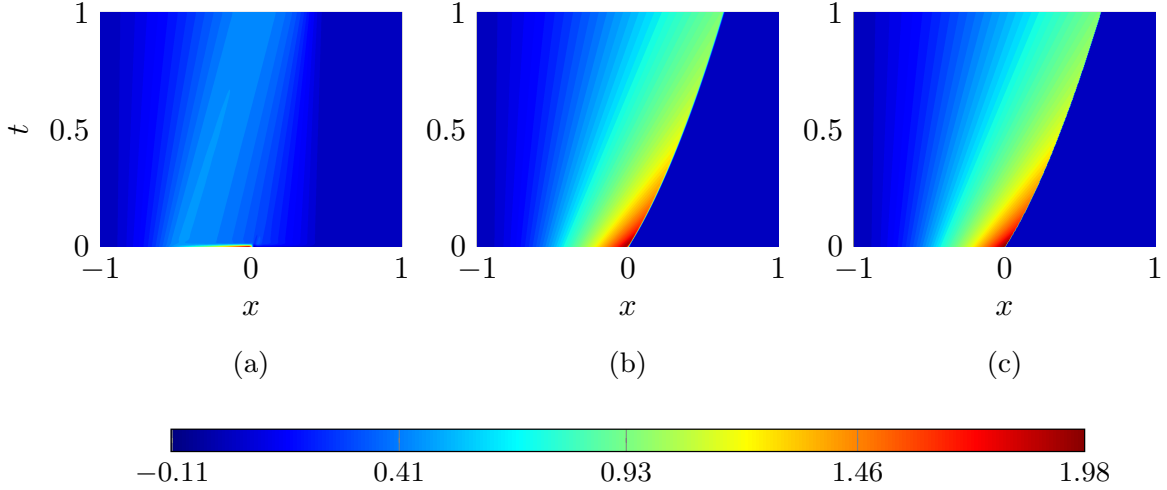

\begin{figure}[h]
  \centering
  \resizebox{0.9\textwidth}{!}{\begin{tikzpicture}
\begin{axis}[
xmin=-1,
xmax=1,
ymin=0,
ymax=1,
xlabel=$x$,
ylabel=$t$,
xtick={-1, -0.5, 0, 0.5, 1},
ytick={0, 0.25, 0.5, 0.75, 1},
width=0.40\textwidth,
height=4.4cm,
name=burgdecp0]
\addplot []
graphics [xmin=-1,xmax=1,ymin=0,ymax=1] { 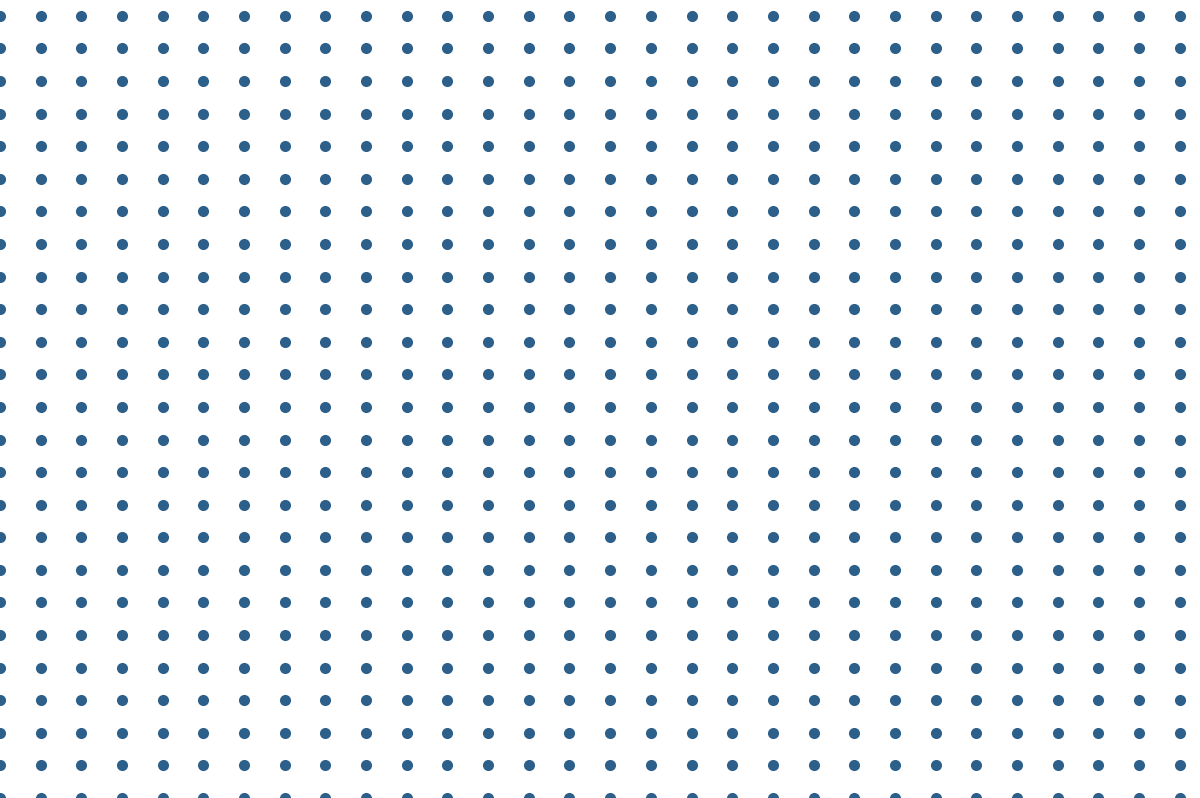};

\end{axis}
\begin{axis}[
xmin=-1,
xmax=1,
ymin=0,
ymax=1,
xlabel=$x$,
ylabel= ,
xtick={-1, -0.5, 0, 0.5, 1},
ytick={0, 0.25, 0.5, 0.75, 1},
width=0.40\textwidth,
height=4.4cm,
name=burgdecp1,
anchor=west,
at={(burgdecp0.east)},
xshift=1.15cm]
\addplot []
graphics [xmin=-1,xmax=1,ymin=0,ymax=1] { 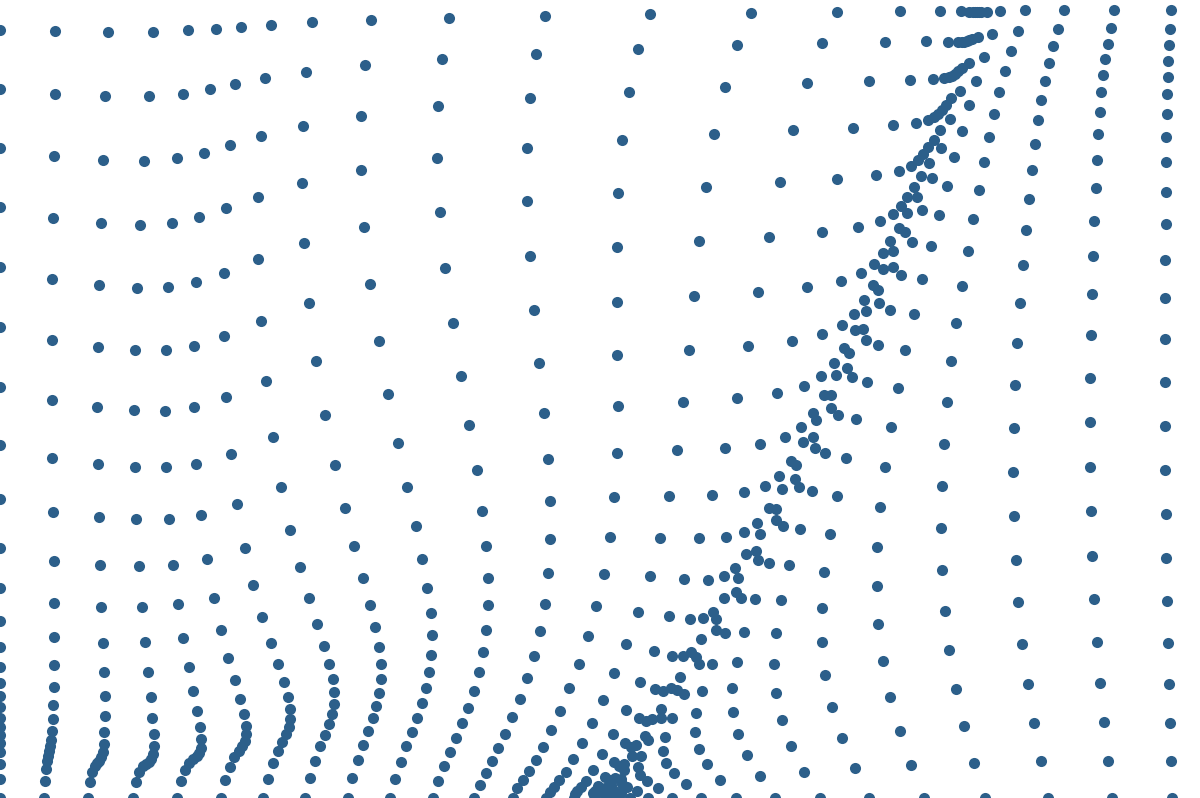};

\end{axis}
\node[anchor=north, font=\small] at ($(burgdecp0.south) + (0, -1.05cm)$) {(a)};
\node[anchor=north, font=\small] at ($(burgdecp1.south) + (0, -1.05cm)$) {(b)};
\end{tikzpicture}}
  \caption{Decelerating Burgers shock: \textbf{(a)} uniform reference
  collocation grid $\mathcal{X}_c^0$ and \textbf{(b)} its image
  $\mathcal{G}(\mathcal{X}_c^0;\boldsymbol{\tau}_*)$ after joint training.
  Points concentrate along the curved space--time trajectory of the
  decelerating shock, a feature a separable tensor-product map cannot
  represent.}
  \label{fig:burgdec_points}
\end{figure}

% =====================================================================
\subsection{One-dimensional Euler shock tube in space-time}
\label{subsec:numexp:sst}

\paragraph{Setup}
We now move to a system of conservation laws, solving the one-dimensional
compressible Euler equations for the Sod shock tube in a space-time setting on
$\Omega=[0,1]\times[0,0.2]$. The conserved state is
$U = [\rho,\ \rho v,\ \rho E]$ with $\rho E = P/(\gamma-1) + \tfrac12\rho v^2$
and $\gamma=1.4$, and the space-time fluxes are
\begin{equation}
    F_I(U) = \begin{bmatrix}
      \rho v & \rho \\
      \rho v^2 + P & \rho v \\
      (\rho E + P)v & \rho E
    \end{bmatrix},
    \qquad
    F_V(U,\nabla U) = \begin{bmatrix}
      \nu\,\partial_{x_1} \rho & 0\\
      \nu\,\partial_{x_1}(\rho v) & 0\\
      \nu\,\partial_{x_1}(\rho E) & 0
    \end{bmatrix},
    \label{eqn:sst-flux}
\end{equation}
where the artificial viscosity $\nu = 8\times10^{-4}$. The outputs of the solution network are the
primitive variables $W = [\rho,\ P,\ v]$ and the conserved state and its
gradients are reconstructed by the chain rule. The initial condition is the
following Riemann problem,
\begin{equation}
    [\rho,P,v](x_1,0) =
    \begin{cases}
      [1.0,\ 1.0,\ 0.0], & x_1<0.5,\\
      [0.125,\ 0.1,\ 0.0], & x_1\ge0.5,
    \end{cases}
    \label{eqn:sod-ic}
\end{equation}
imposed on the face $t=0$. No spatial boundary conditions are required because
none of the waves reaches $x_1=0$ or $x_1=1$ within the time window. The solution
consists of three waves emanating from $(x_1,t)=(0.5,0)$---a left-moving
rarefaction fan, a contact discontinuity, and a right-moving shock.

\paragraph{Exact solution}
The reference is the exact solution of the Riemann problem, obtained by solving
the star-region pressure equation to machine precision by Newton iteration and
then sampling the resulting wave structure along $\xi=(x_1-0.5)/t$.

\paragraph{Results}
Figure~\ref{fig:sst_slices} compares $\rho$, $P$, and $v$ against the exact
solution at $t=0.05,0.10,0.15,0.20$. FT-PINN is able to approximate the rarefaction fan, 
contact discontinuity, and the shock well at all four times. The vanilla PINN, on the other hand, 
is unable to resolve any of the three waves. Furthermore, Figures~\ref{fig:sst_rho_contours} and~\ref{fig:sst_points} 
show the contours of the density field over the entire domain and the collocation distribution, respectively. The 
mapped points fan out from the membrane separating the high and low pressure regions, but they do not distribute themselves
evenly among the three waves. The points concentrate strongly along the shock and, to a lesser extent, along
the rarefaction, while the contact receives negligible dedicated clustering. Since the deformation
parameters are driven only by the pulled-back residual, the map clusters points only where the residual
is large, and the residual of the trained FT-PINN shows that this is not the case at the contact. Using the
exact solution to locate the waves, we compute the root-mean-square residual norm in a band of half-width
$0.02$ in $x_1$ about the contact and about the shock for $t\ge0.05$, where the bands are disjoint, as well as
over the smooth regions away from all three waves. The residual is $1.8\times10^{-3}$ in the contact band,
$1.2\times10^{-2}$ in the shock band, and $7.1\times10^{-3}$ in the smooth regions. The contact is therefore
indistinguishable from smooth flow as far as the loss is concerned, so there is nothing for the map to gain by moving points toward it. Despite receiving no
dedicated points, the contact is nevertheless captured well in
Figure~\ref{fig:sst_slices}, consistent with its low residual.
\begin{figure}[h]
  \centering
  \resizebox{0.98\textwidth}{!}{\input{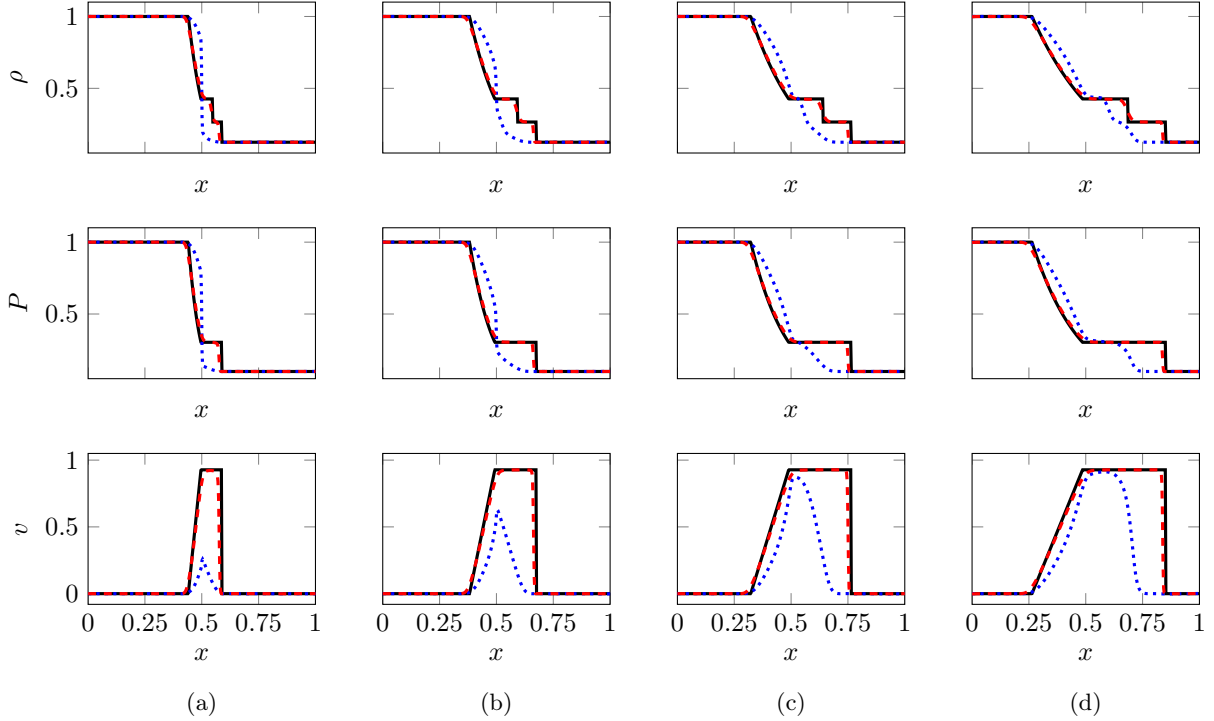}}
  \caption{Euler shock tube: primitive profiles against $x_1$, with $\rho$ in
  the top row, $P$ in the middle row, and $v$ in the bottom row. Columns are
  \textbf{(a)} $t=0.05$, \textbf{(b)} $t=0.10$, \textbf{(c)} $t=0.15$, and
  \textbf{(d)} $t=0.20$, comparing the exact Riemann solution
  (\ref{line:sst_exact}), the vanilla PINN (\ref{line:sst_A}), and FT-PINN
  (\ref{line:sst_B}).}
  \label{fig:sst_slices}
\end{figure}

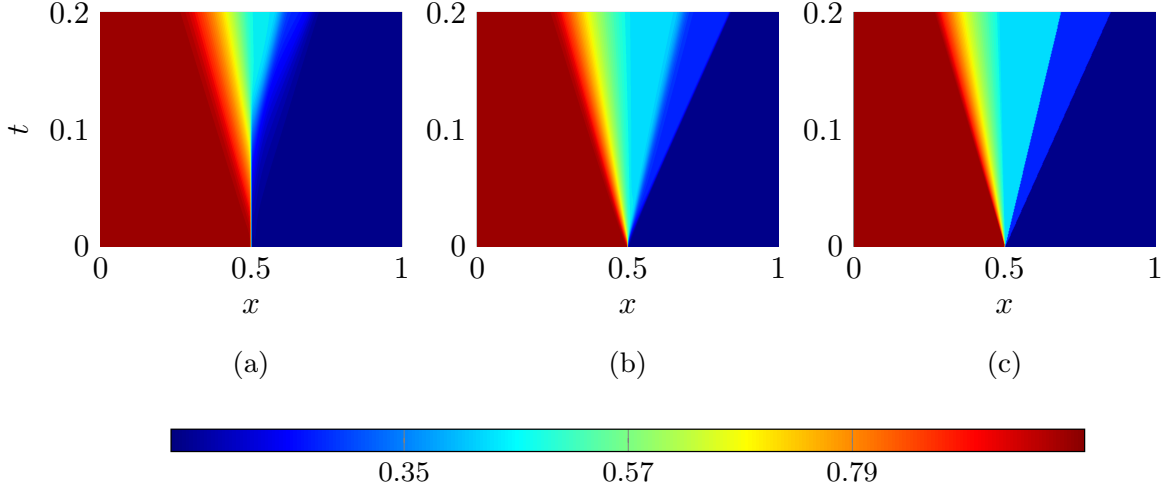
\begin{figure}[h]
  \centering
  \resizebox{0.95\textwidth}{!}{\begin{tikzpicture}
\begin{axis}[
xmin=0,
xmax=1,
ymin=0,
ymax=0.2,
xlabel=$x$,
ylabel=$t$,
xtick={0, 0.5, 1},
ytick={0, 0.1, 0.2},
width=0.30\textwidth,
height=4.2cm,
name=sst_rhoc0]
\addplot []
graphics [xmin=0,xmax=1,ymin=0,ymax=0.2] { 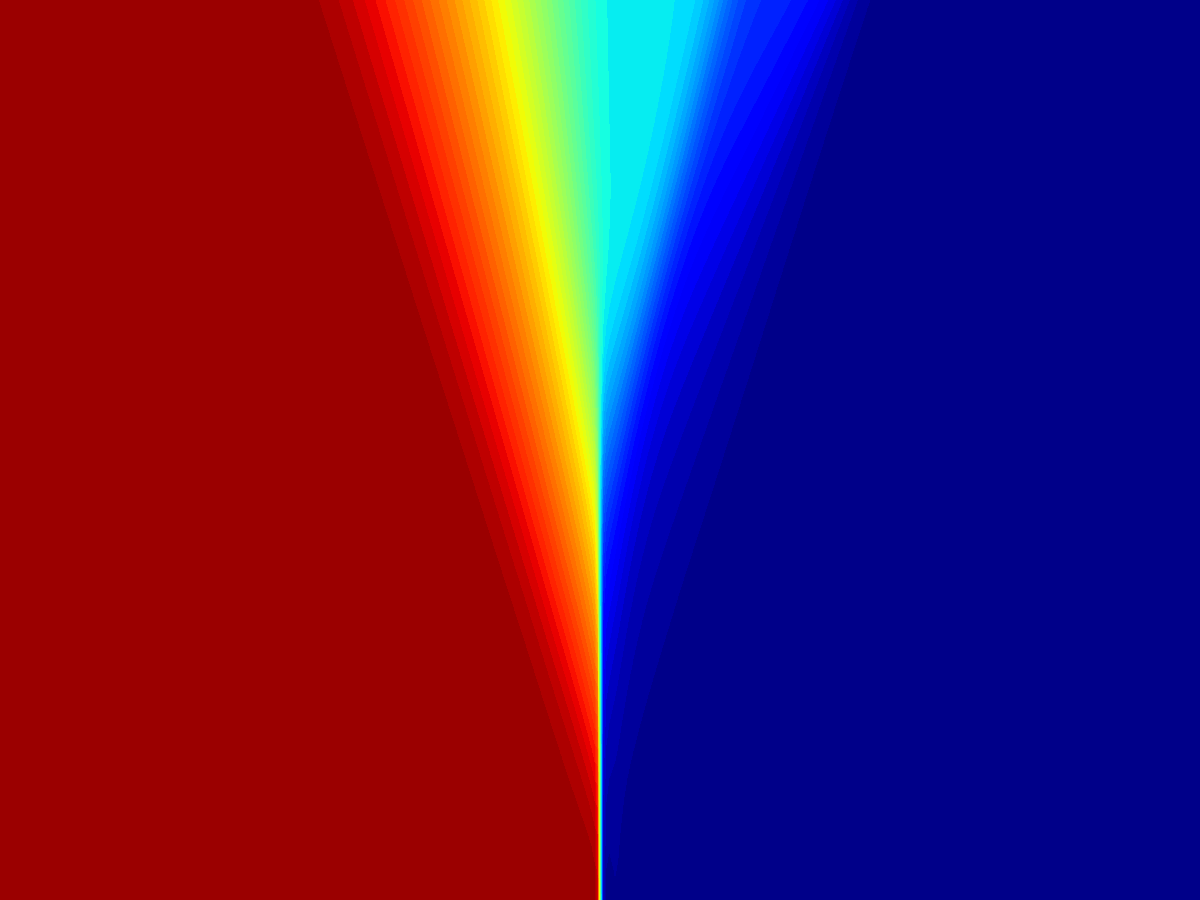};

\end{axis}
\begin{axis}[
xmin=0,
xmax=1,
ymin=0,
ymax=0.2,
xlabel=$x$,
ylabel= ,
xtick={0, 0.5, 1},
ytick={0, 0.1, 0.2},
width=0.30\textwidth,
height=4.2cm,
name=sst_rhoc1,
anchor=west,
at={(sst_rhoc0.east)},
xshift=0.85cm]
\addplot []
graphics [xmin=0,xmax=1,ymin=0,ymax=0.2] { 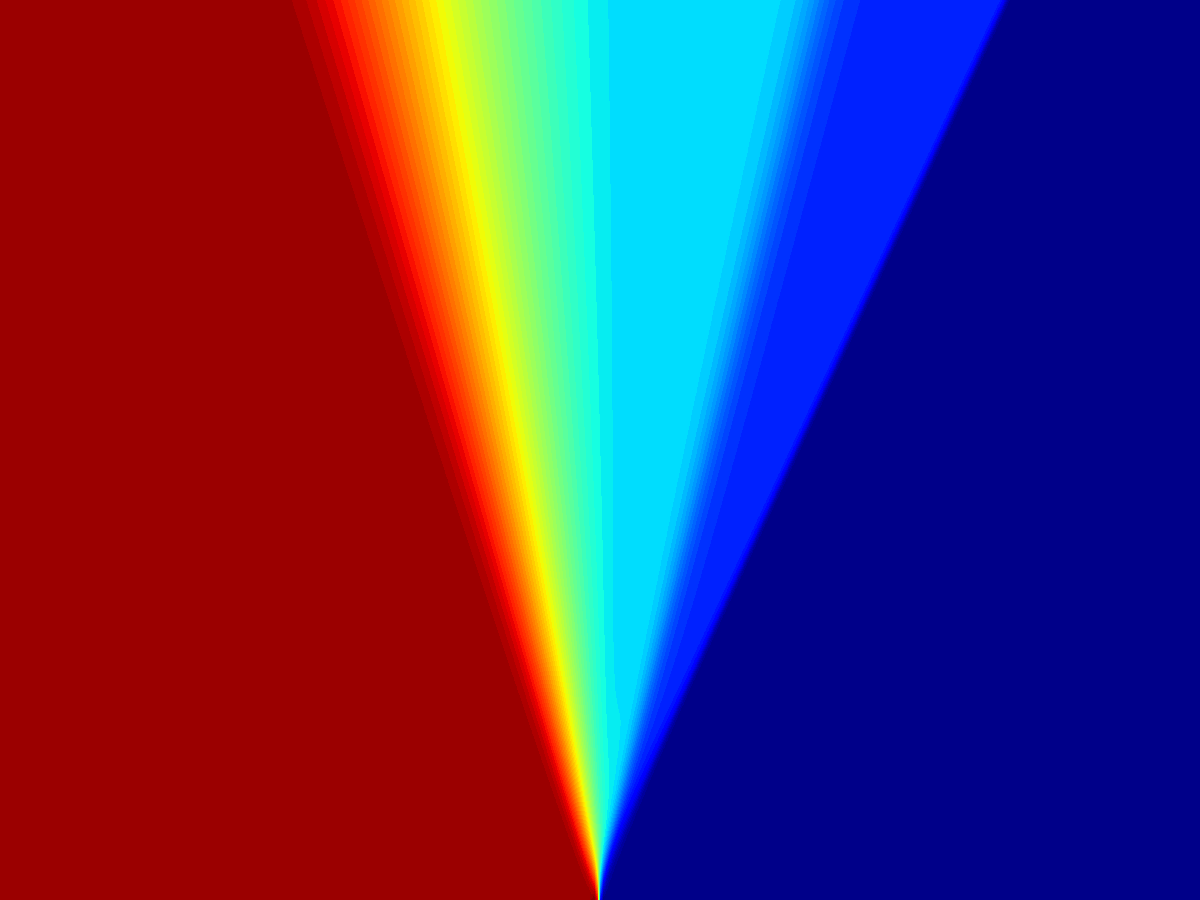};

\end{axis}
\begin{axis}[
xmin=0,
xmax=1,
ymin=0,
ymax=0.2,
xlabel=$x$,
ylabel= ,
xtick={0, 0.5, 1},
ytick={0, 0.1, 0.2},
width=0.30\textwidth,
height=4.2cm,
name=sst_rhoc2,
anchor=west,
at={(sst_rhoc1.east)},
xshift=0.85cm]
\addplot []
graphics [xmin=0,xmax=1,ymin=0,ymax=0.2] { 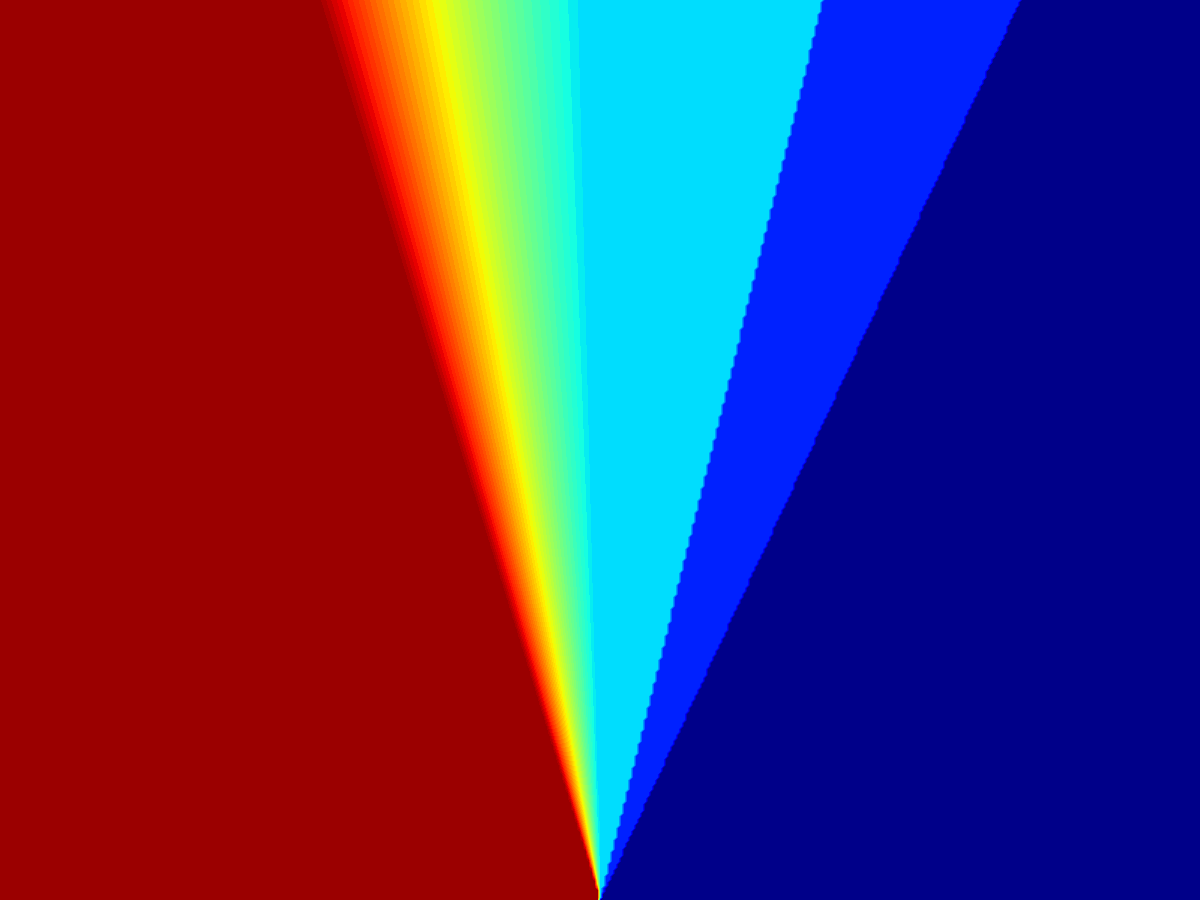};

\end{axis}
\node[anchor=north, font=\small] at ($(sst_rhoc0.south) + (0, -1.05cm)$) {(a)};
\node[anchor=north, font=\small] at ($(sst_rhoc1.south) + (0, -1.05cm)$) {(b)};
\node[anchor=north, font=\small] at ($(sst_rhoc2.south) + (0, -1.05cm)$) {(c)};
\begin{axis}[
hide axis, scale only axis,
xmin=0, xmax=1, ymin=0, ymax=1,
height=0pt, width=0pt,
colormap/jet,
colorbar horizontal,
point meta min=0.1212, point meta max=1.018,
colorbar style={
  width=0.62\textwidth, height=0.25cm,
  xtick={0.12, 0.35, 0.57, 0.79, 1.02},
  ticklabel style={font=\small},
  at={(sst_rhoc1.south)}, anchor=north, yshift=-1.75cm
}]
\addplot [draw=none] coordinates {(0,0)};
\end{axis}

\end{tikzpicture}}
  \caption{Euler shock tube: contours of density $\rho$ over the space--time
  slab for \textbf{(a)} the vanilla PINN, \textbf{(b)} FT-PINN, and
  \textbf{(c)} the exact Riemann solution.}
  \label{fig:sst_rho_contours}
\end{figure}

\begin{figure}[h]
  \centering
  \resizebox{0.9\textwidth}{!}{\begin{tikzpicture}
\begin{axis}[
xmin=0,
xmax=1,
ymin=0,
ymax=0.2,
xlabel=$x$,
ylabel=$t$,
xtick={0, 0.25, 0.5, 0.75, 1},
ytick={0, 0.05, 0.1, 0.15, 0.2},
width=0.40\textwidth,
height=4.4cm,
name=sstp0,
yticklabels={0, 0.05, 0.1, 0.15, 0.2}]
\addplot []
graphics [xmin=0,xmax=1,ymin=0,ymax=0.2] { 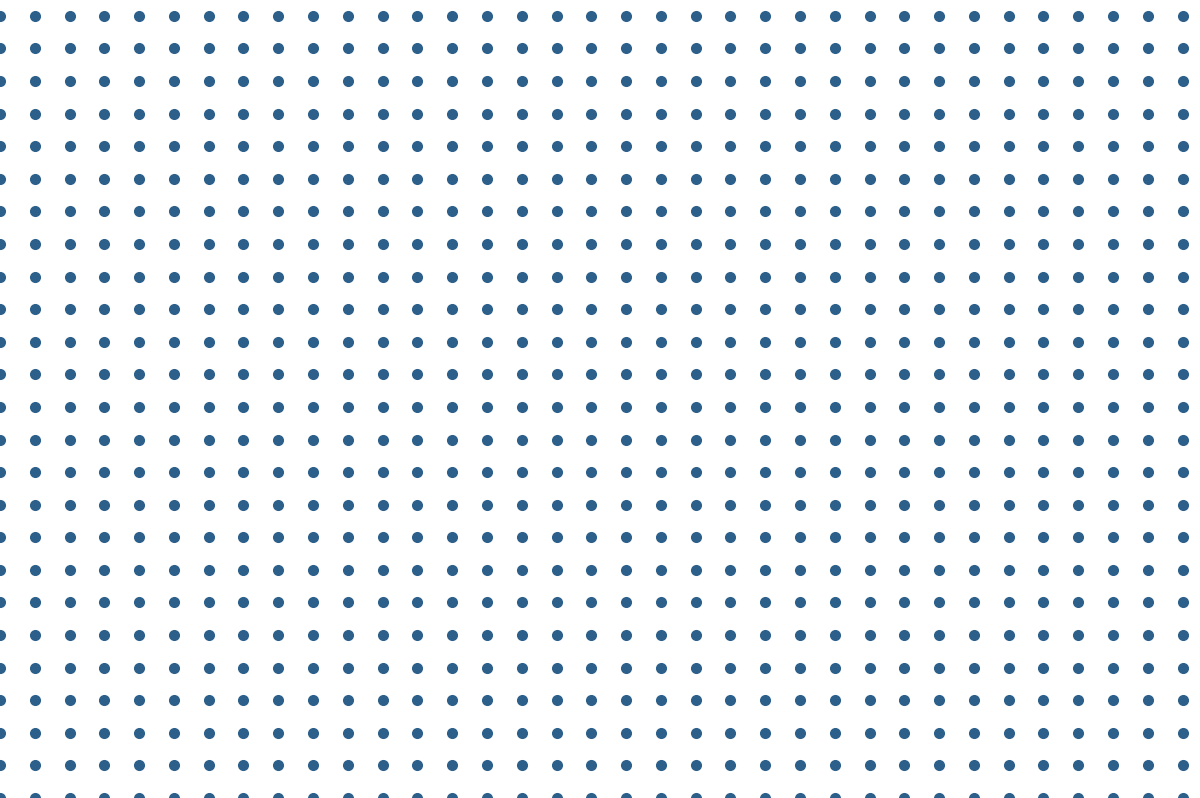};

\end{axis}
\begin{axis}[
xmin=0,
xmax=1,
ymin=0,
ymax=0.2,
xlabel=$x$,
ylabel= ,
xtick={0, 0.25, 0.5, 0.75, 1},
ytick={0, 0.05, 0.1, 0.15, 0.2},
width=0.40\textwidth,
height=4.4cm,
name=sstp1,
yticklabels={0, 0.05, 0.1, 0.15, 0.2},
anchor=west,
at={(sstp0.east)},
xshift=1.15cm]
\addplot []
graphics [xmin=0,xmax=1,ymin=0,ymax=0.2] { 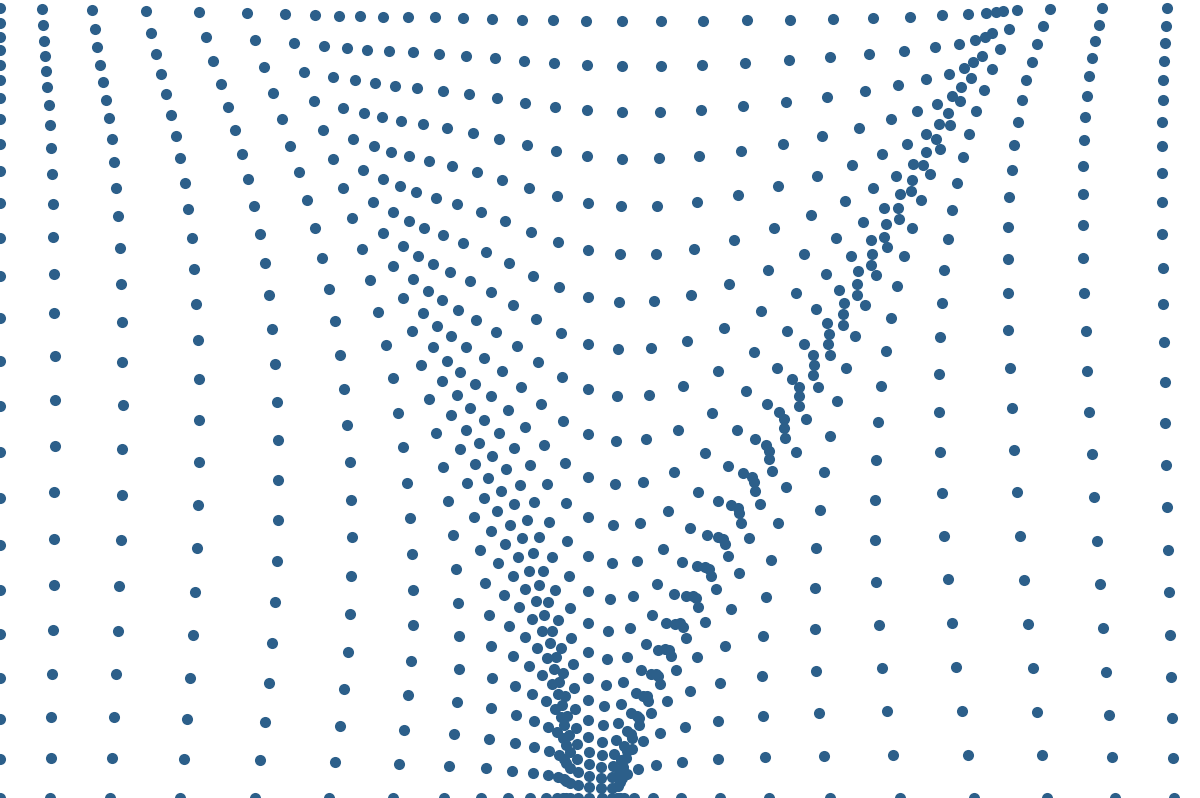};

\end{axis}
\node[anchor=north, font=\small] at ($(sstp0.south) + (0, -1.05cm)$) {(a)};
\node[anchor=north, font=\small] at ($(sstp1.south) + (0, -1.05cm)$) {(b)};
\end{tikzpicture}}
  \caption{Euler shock tube: \textbf{(a)} uniform reference collocation grid
  $\mathcal{X}_c^0$ and \textbf{(b)} its image
  $\mathcal{G}(\mathcal{X}_c^0;\boldsymbol{\tau}_*)$ after joint training.}
  \label{fig:sst_points}
\end{figure}

% =====================================================================
\subsection{Steady two-dimensional Euler regular shock reflection}
\label{subsec:numexp:rr}

\paragraph{Setup}
Finally we consider a spatially two-dimensional steady problem with two
oblique shocks. We solve the steady compressible Euler equations on the
rectangle $\Omega=[0,2]\times[0,1]$ with the conserved state
$U = [\rho,\ \rho u,\ \rho v,\ \rho E]$ and the inviscid flux
\begin{equation}
    F_I(U) = \begin{bmatrix}
      \rho u & \rho v\\
      \rho u^2 + P & \rho u v\\
      \rho u v & \rho v^2 + P\\
      (\rho E+P)u & (\rho E+P)v
    \end{bmatrix},
    \label{eqn:rr-flux}
\end{equation}
regularized by artificial viscous flux $F_V = \nu\nabla U$ with
$\nu = 7.5\times10^{-4}$. The network outputs the primitives
$W=[\rho,u,v,P]$. The inflow is supersonic, and an incident oblique shock is
supplied ready-made through the inlet by a two-state Dirichlet profile on
$x_1=0$. The uniform pre-shock freestream state $W_1$, at $M_1 = 2.9$ in the
$+x_1$ direction, is imposed for $x_2 \ge x_{2,\mathrm{split}}$ and the
post-shock state $W_2$ for $x_2 < x_{2,\mathrm{split}}$, with
$x_{2,\mathrm{split}}=0.5$. Across the incident shock the flow turns upward by
the deflection angle $\theta_w = 10^\circ$, so the shock runs obliquely toward
the top wall and impinges on it at $x_{1,\mathrm{imp}} = 0.935$
(\ref{app:rr}; the configuration and all angles are sketched to scale in
Figure~\ref{fig:rr_schematic}). A slip wall boundary condition
is applied at the face $x_2=1$. No boundary conditions are imposed on the faces $x_2=0$ and $x_1=2$.
The deflection angle $\theta_w$ is small enough that a single reflected shock,
oblique and inclined at a non-specular angle to the wall, turns the stream back parallel to the wall.
The resulting two-shock configuration is termed a regular reflection \cite{bendor2007shock}.

\paragraph{Exact solution}
The complete analytical solution consists of three uniform states separated by
two straight oblique shocks, and is obtained by applying the oblique-shock
relations across each leg in sequence. The full derivation,
from the Rankine--Hugoniot conditions through the $\theta$--$\beta$--$M$
relation~\eqref{eqn:app-rr-tbm}, is given in \ref{app:rr}. Applying the
relation once at $(M_1,\theta_w)$ gives the incident shock angle
$\beta = 28.13^\circ$ and the post-incident state, and applying it a second
time at
$(M_2,\theta_w)$---the reflected shock must turn the stream, inclined at
$+\theta_w$, back to wall-parallel---gives
$\beta_{\mathrm{refl}} = 32.74^\circ$ and hence the reflected leg inclination
$\phi_{\mathrm{refl}} = \beta_{\mathrm{refl}} - \theta_w
= 22.74^\circ$ measured from the wall. This quantifies the non-specularity
noted above, i.e., $\phi_{\mathrm{refl}} \ne \beta$.
The three states are
\begin{equation}
\begin{aligned}
  W_1 &= (1.000,\ 3.431,\ 0.000,\ 1.000), \quad M_1 = 2.90,\\
  W_2 &= (1.633,\ 3.136,\ 0.553,\ 2.014), \quad M_2 = 2.42,\\
  W_3 &= (2.505,\ 2.904,\ 0.000,\ 3.700), \quad M_3 = 2.02,
\end{aligned}
    \label{eqn:rr-states}
\end{equation}
in the ordering $(\rho,u,v,P)$, with $v_3=0$ confirming that the reflected shock
restores wall-parallel flow.

\paragraph{A second vanilla PINN with random collocation points}
On this problem the vanilla PINN trained on the uniform $98\times58$ grid of
Table~\ref{tab:hyper} collapses to a nearly uniform field. The interior grid
excludes the faces of $\Omega$, so its first column lies at $x_1\approx0.02$,
and the soft-constrained network places the entire inlet jump in the band
between the inlet face and this column, where no residual is evaluated, after
which a nearly constant field drives the residual to a small value everywhere
else. To show what is needed for the vanilla PINN to produce a shock
structure at all, we also report a vanilla PINN trained on $N_c=12{,}000$
interior collocation points drawn once, uniformly at random over $\Omega$
with a fixed seed, and held fixed for the whole run (Figure~\ref{fig:rr_points}(c)).
A random set has no empty band along the faces, and the budget is roughly twice
the $N_c=5684$ interior points of the grid. The batch size is doubled to $N_B=6000$ so that the number of Adam
steps is unchanged, and all other settings, including the boundary point sets,
the network, and the training horizon, are those of Table~\ref{tab:hyper}. FT-PINN
is trained on the $98\times58$ grid as before.

\paragraph{Results}
Figure~\ref{fig:rr_p_contours} compares pressure contours against the
analytical three-state solution. The vanilla PINN on the uniform grid produces an essentially uniform field,
having failed to develop either shock, so the incident shock injected at the
inlet is immediately diffused away and no reflection forms at all. With
$12{,}000$ random collocation points the vanilla PINN does develop the two-shock
pattern, but both legs are broadened to many times the viscous width, the
incident leg in particular, the reflection point on the wall is poorly
defined, and the post-reflection state behind the reflected leg is not
uniform. FT-PINN, on the same $98\times58$ grid on which the first vanilla PINN
collapsed, recovers the full two-shock pattern---an incident leg at the correct angle,
a well-defined reflection point on the wall near
$x_{1,\mathrm{imp}}=0.935$, and a reflected leg at the correct non-specular
inclination---in close agreement with the analytical field.
Figure~\ref{fig:rr_slices} shows vertical slices for $P$ and $\rho$ at
$x_1 = 0.4, 0.7$ (crossing the incident leg) and $x_1 = 1.2, 1.6$ (crossing the
reflected leg). FT-PINN closely approximates the near-vertical jumps. The vanilla PINN
on the grid shows no jump structure, while the random-point vanilla PINN
reaches roughly the correct post-shock levels but spreads each jump over a
large fraction of the channel height.
Figure~\ref{fig:rr_points} shows the collocation distribution of FT-PINN. The points get packed along
both shock legs, meeting at the wall impingement point, and slide along the
faces where the shocks meet the boundary---the inlet, the wall, and the
outflow---without leaving the domain.

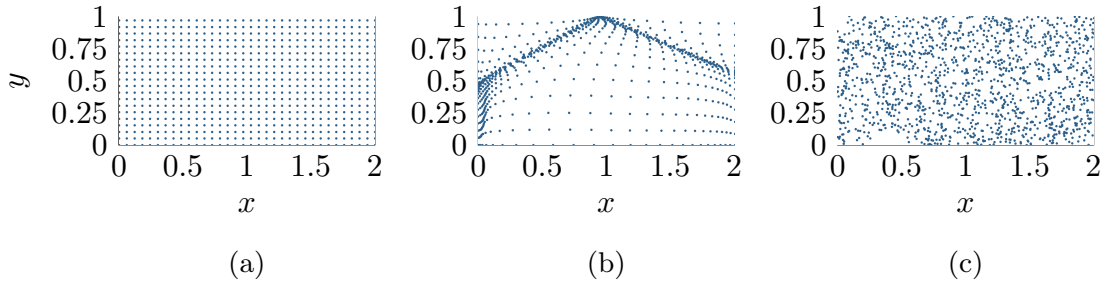
\begin{figure}[h]
  \centering
  \resizebox{0.9\textwidth}{!}{\begin{tikzpicture}
\begin{axis}[
axis equal image,
xmin=0,
xmax=2,
ymin=0,
ymax=1,
xlabel=$x$,
ylabel=$y$,
xtick={0, 0.5, 1, 1.5, 2},
ytick={0, 0.25, 0.5, 0.75, 1},
width=0.30\textwidth,
height=3.0cm,
name=rrp0]
\addplot []
graphics [xmin=0,xmax=2,ymin=0,ymax=1] { 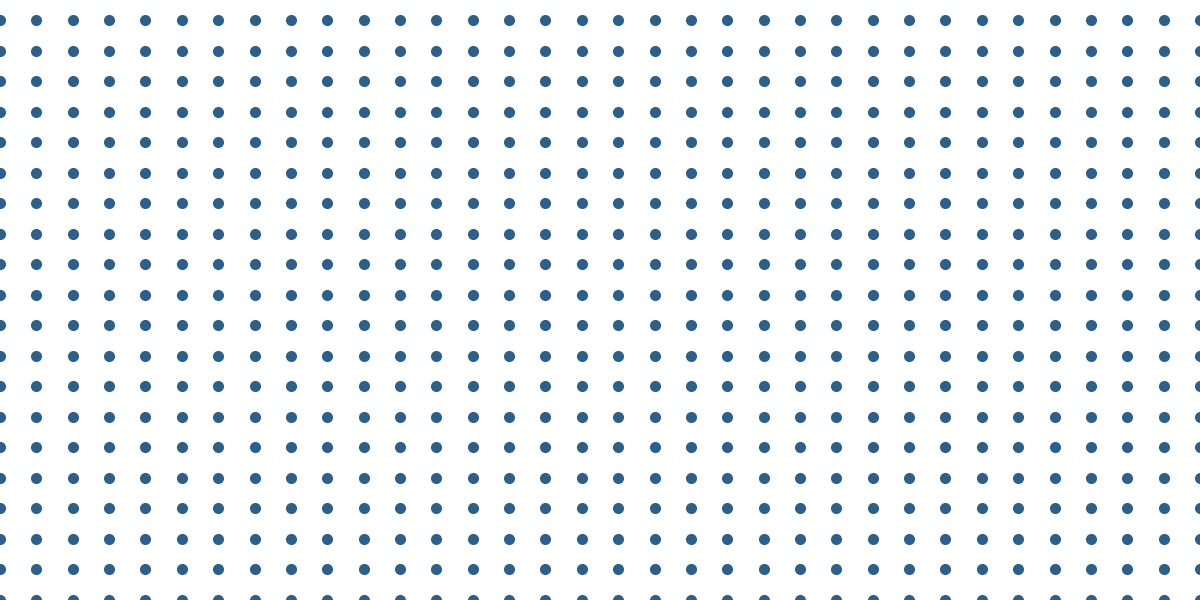};

\end{axis}
\begin{axis}[
axis equal image,
xmin=0,
xmax=2,
ymin=0,
ymax=1,
xlabel=$x$,
ylabel= ,
xtick={0, 0.5, 1, 1.5, 2},
ytick={0, 0.25, 0.5, 0.75, 1},
width=0.30\textwidth,
height=3.0cm,
name=rrp1,
anchor=west,
at={(rrp0.east)},
xshift=1.15cm]
\addplot []
graphics [xmin=0,xmax=2,ymin=0,ymax=1] { 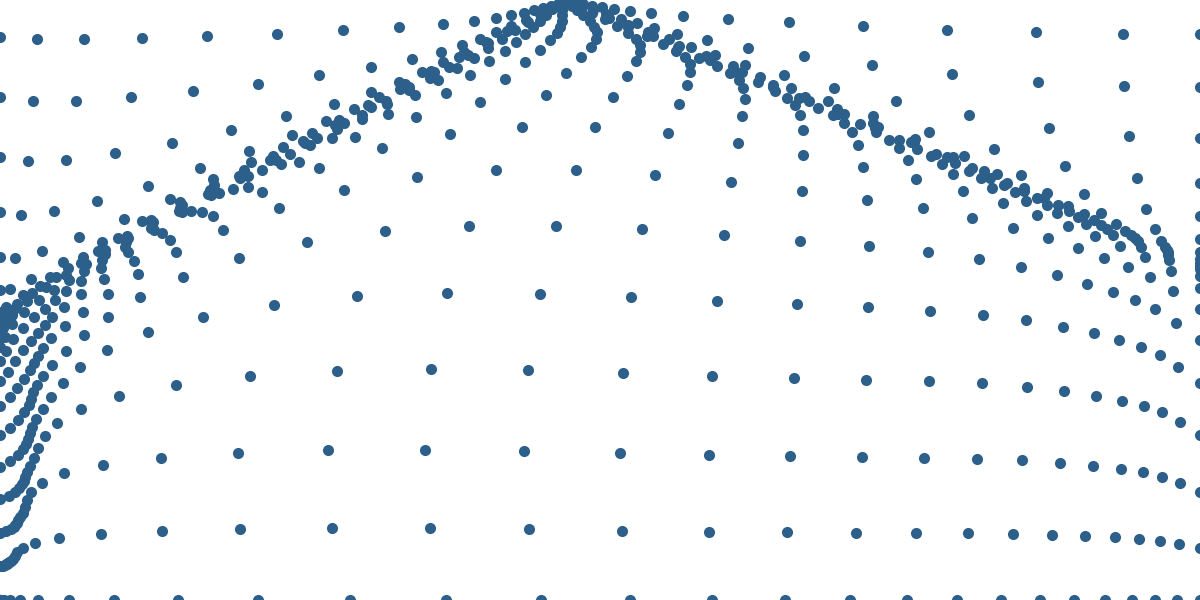};

\end{axis}
\begin{axis}[
axis equal image,
xmin=0,
xmax=2,
ymin=0,
ymax=1,
xlabel=$x$,
ylabel= ,
xtick={0, 0.5, 1, 1.5, 2},
ytick={0, 0.25, 0.5, 0.75, 1},
width=0.30\textwidth,
height=3.0cm,
name=rrp2,
anchor=west,
at={(rrp1.east)},
xshift=1.15cm]
\addplot []
graphics [xmin=0,xmax=2,ymin=0,ymax=1] { 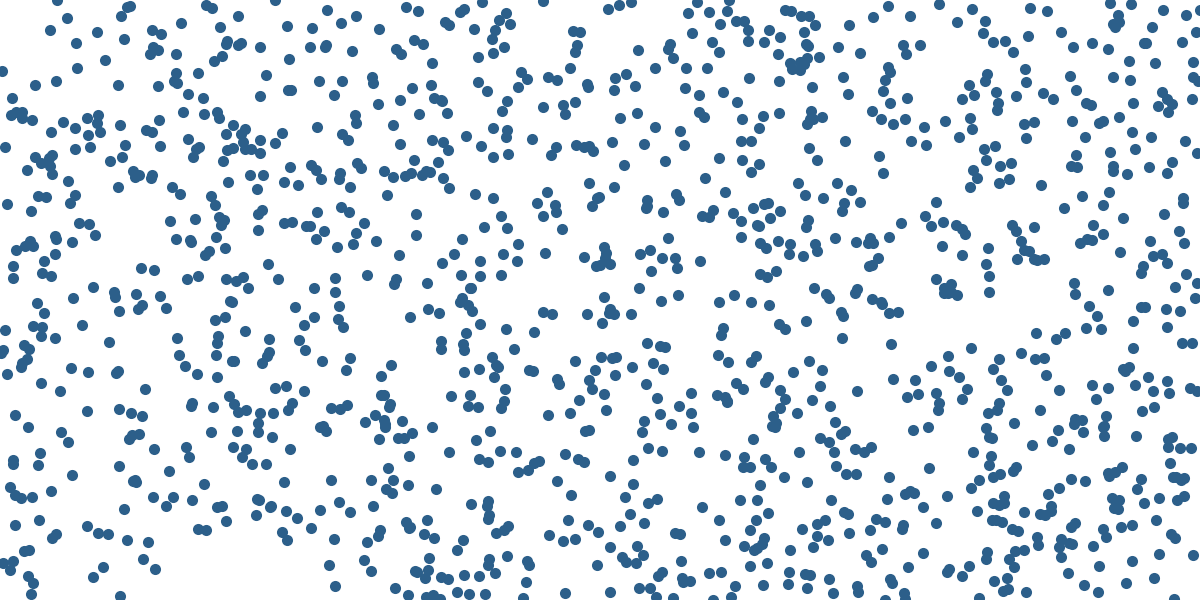};

\end{axis}
\node[anchor=north, font=\small] at ($(rrp0.south) + (0, -1.05cm)$) {(a)};
\node[anchor=north, font=\small] at ($(rrp1.south) + (0, -1.05cm)$) {(b)};
\node[anchor=north, font=\small] at ($(rrp2.south) + (0, -1.05cm)$) {(c)};
\end{tikzpicture}}
  \caption{Regular shock reflection: \textbf{(a)} uniform reference collocation
  grid $\mathcal{X}_c^0$ and \textbf{(b)} its image
  $\mathcal{G}(\mathcal{X}_c^0;\boldsymbol{\tau}_*)$ after joint training;
  \textbf{(c)} the fixed set of $12{,}000$ uniformly random collocation points
  used by the second vanilla PINN. The same fraction of points is shown in
  every panel.}
  \label{fig:rr_points}
\end{figure}

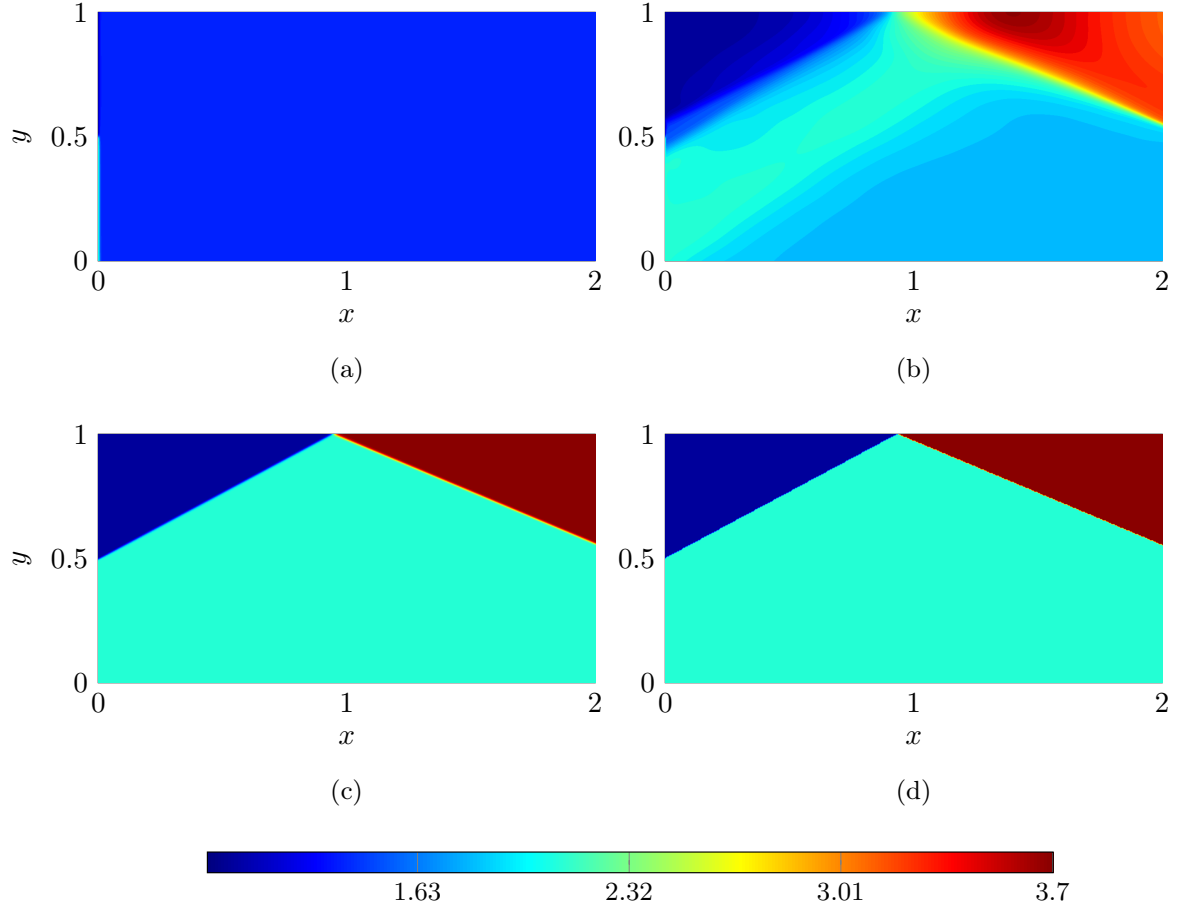
\begin{figure}[h]
  \centering
  \resizebox{0.95\textwidth}{!}{\begin{tikzpicture}
\begin{axis}[
axis equal image,
xmin=0,
xmax=2,
ymin=0,
ymax=1,
xlabel=$x$,
ylabel=$y$,
xtick={0, 1, 2},
ytick={0, 0.5, 1},
width=0.46\textwidth,
height=6cm,
name=rr_pc0]
\addplot []
graphics [xmin=0,xmax=2,ymin=0,ymax=1] { 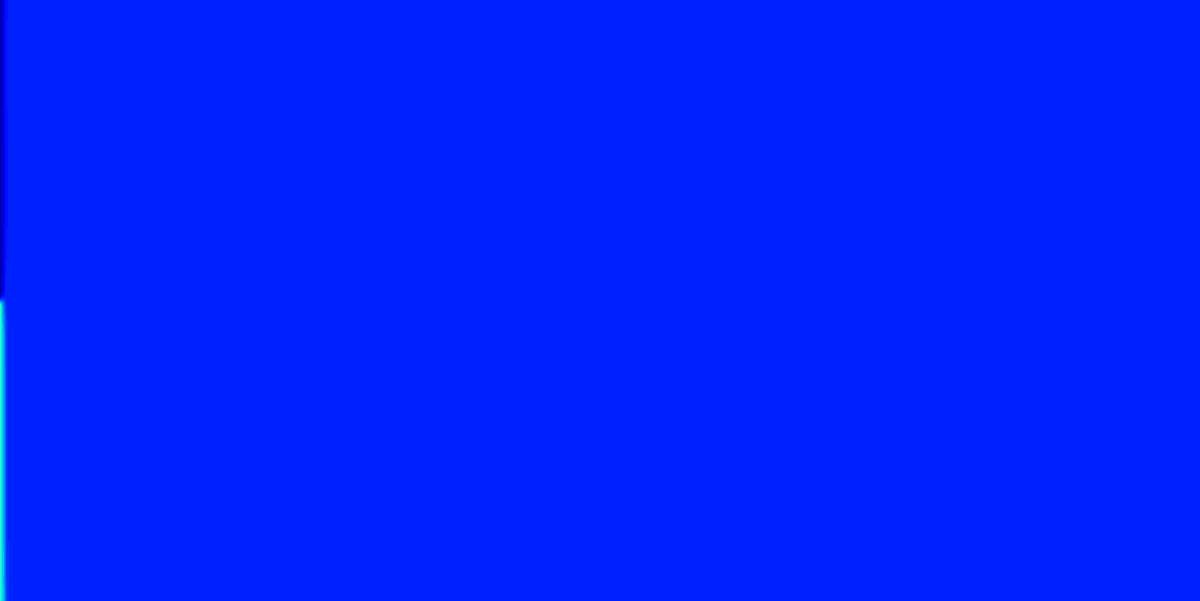};

\end{axis}
\begin{axis}[
axis equal image,
xmin=0,
xmax=2,
ymin=0,
ymax=1,
xlabel=$x$,
ylabel= ,
xtick={0, 1, 2},
ytick={0, 0.5, 1},
width=0.46\textwidth,
height=6cm,
name=rr_pc1,
anchor=west,
at={(rr_pc0.east)},
xshift=0.85cm]
\addplot []
graphics [xmin=0,xmax=2,ymin=0,ymax=1] { 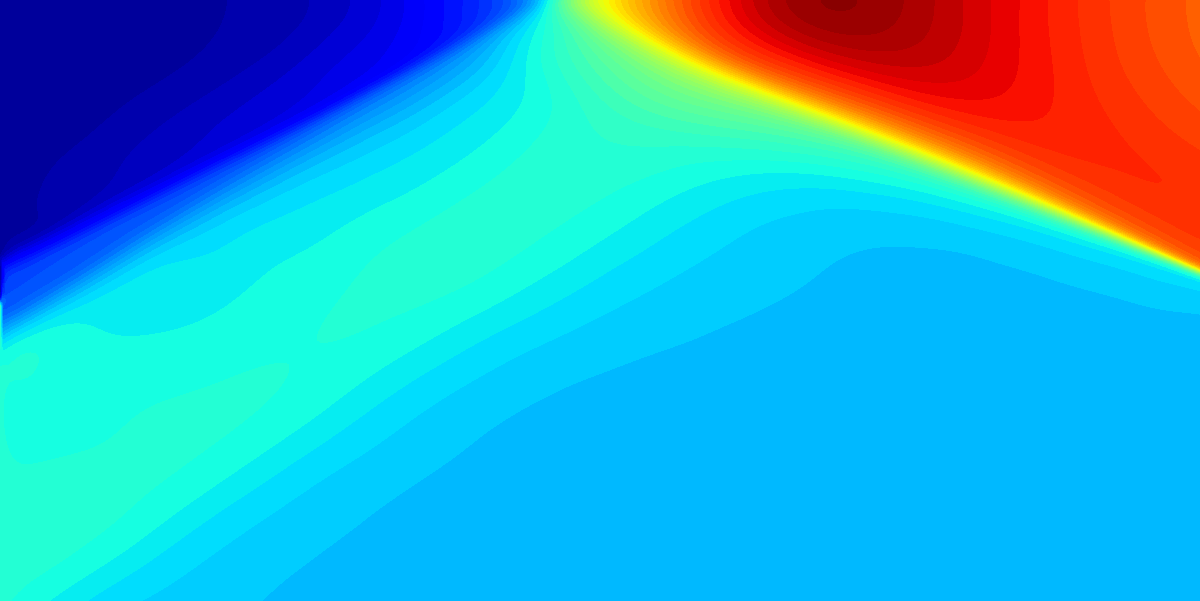};

\end{axis}
\begin{axis}[
axis equal image,
xmin=0,
xmax=2,
ymin=0,
ymax=1,
xlabel=$x$,
ylabel=$y$,
xtick={0, 1, 2},
ytick={0, 0.5, 1},
width=0.46\textwidth,
height=6cm,
name=rr_pc2,
anchor=north,
at={(rr_pc0.south)},
yshift=-2.1cm]
\addplot []
graphics [xmin=0,xmax=2,ymin=0,ymax=1] { 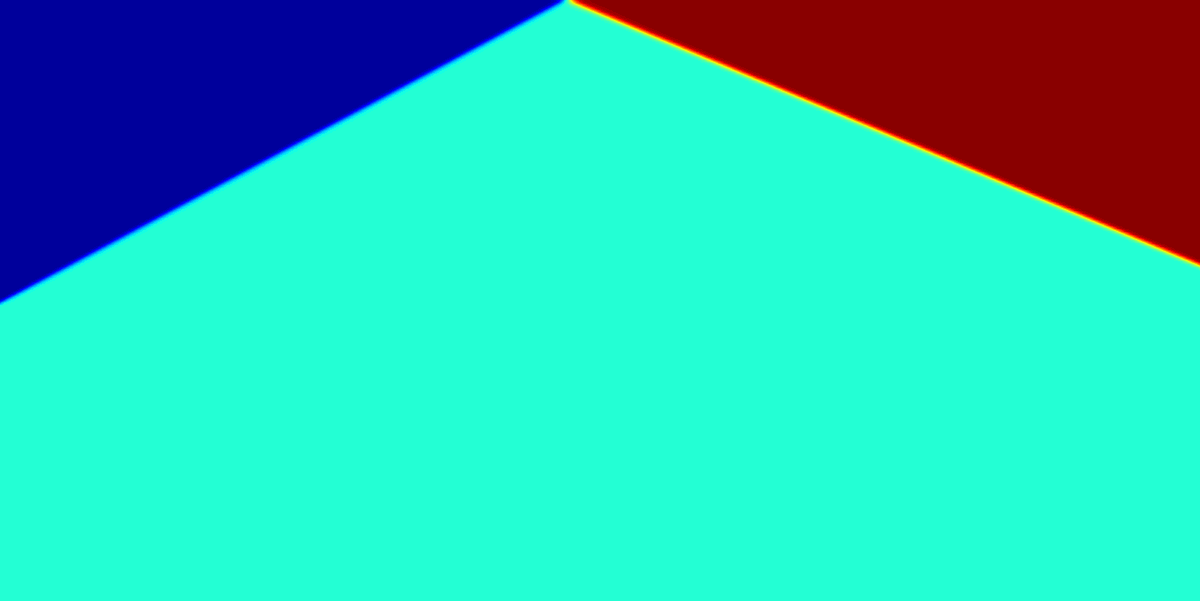};

\end{axis}
\begin{axis}[
axis equal image,
xmin=0,
xmax=2,
ymin=0,
ymax=1,
xlabel=$x$,
ylabel= ,
xtick={0, 1, 2},
ytick={0, 0.5, 1},
width=0.46\textwidth,
height=6cm,
name=rr_pc3,
anchor=west,
at={(rr_pc2.east)},
xshift=0.85cm]
\addplot []
graphics [xmin=0,xmax=2,ymin=0,ymax=1] { 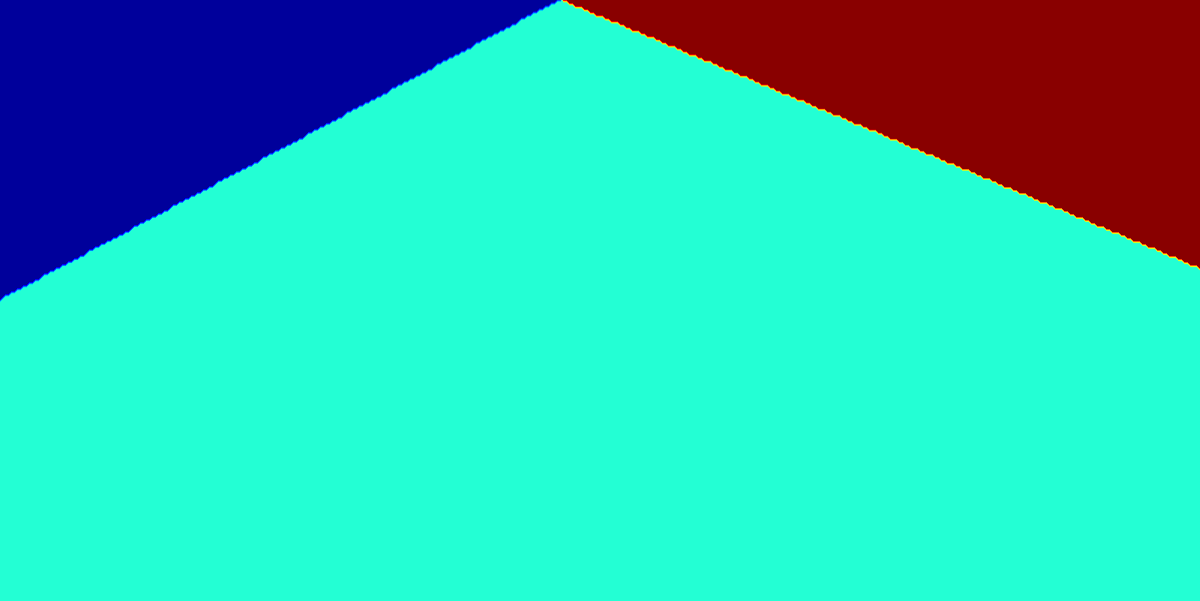};

\end{axis}
\node[anchor=north, font=\small] at ($(rr_pc0.south) + (0, -1.05cm)$) {(a)};
\node[anchor=north, font=\small] at ($(rr_pc1.south) + (0, -1.05cm)$) {(b)};
\node[anchor=north, font=\small] at ($(rr_pc2.south) + (0, -1.05cm)$) {(c)};
\node[anchor=north, font=\small] at ($(rr_pc3.south) + (0, -1.05cm)$) {(d)};
\begin{axis}[
hide axis, scale only axis,
xmin=0, xmax=1, ymin=0, ymax=1,
height=0pt, width=0pt,
colormap/jet,
colorbar horizontal,
point meta min=0.9449, point meta max=3.702,
colorbar style={
  width=0.62\textwidth, height=0.25cm,
  xtick={0.94, 1.63, 2.32, 3.01, 3.70},
  ticklabel style={font=\small},
  at={($(rr_pc2.south)!0.5!(rr_pc3.south)$)}, anchor=north, yshift=-1.75cm
}]
\addplot [draw=none] coordinates {(0,0)};
\end{axis}

\end{tikzpicture}}
  \caption{Regular shock reflection: pressure contours for \textbf{(a)} the
  vanilla PINN on the uniform $98\times58$ grid, \textbf{(b)} the vanilla
  PINN on $12{,}000$ fixed random collocation points, \textbf{(c)} FT-PINN,
  and \textbf{(d)} the analytical solution (\ref{app:rr}).}
  \label{fig:rr_p_contours}
\end{figure}

\begin{figure}[h]
  \centering
  \resizebox{0.98\textwidth}{!}{\input{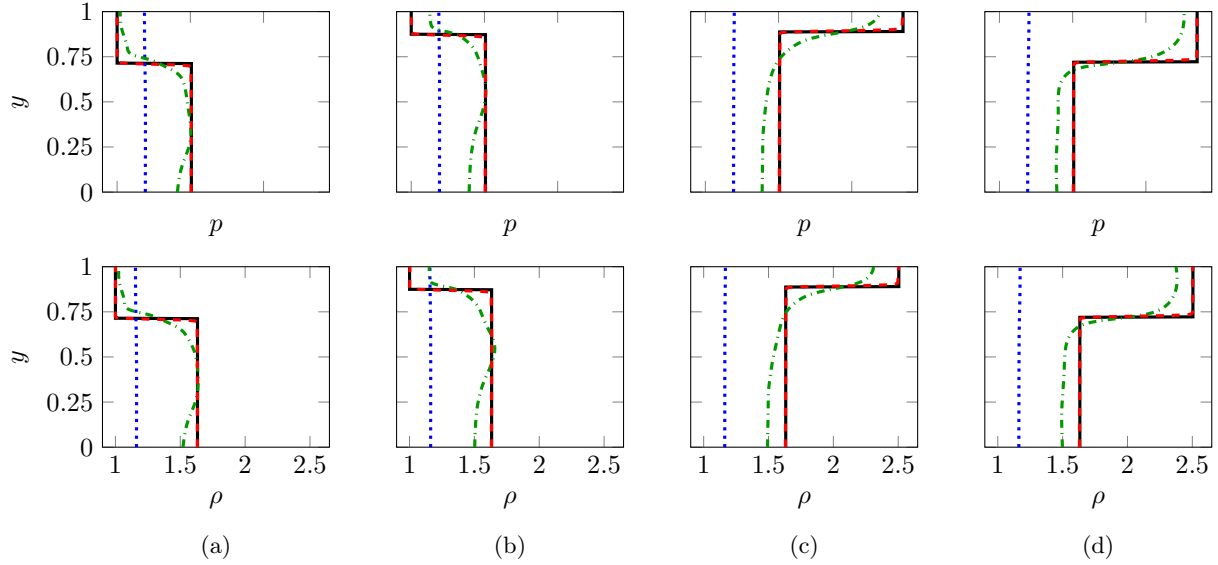}}
  \caption{Regular shock reflection: vertical slices for $P$ in the top row and
  $\rho$ in the bottom row. Columns are \textbf{(a)} $x_1=0.4$ and
  \textbf{(b)} $x_1=0.7$, both crossing the incident leg, and \textbf{(c)}
  $x_1=1.2$ and \textbf{(d)} $x_1=1.6$, both crossing the reflected leg,
  comparing the analytical solution (\ref{line:rr_exact}), the vanilla PINN
  on the uniform grid (\ref{line:rr_A}), the vanilla PINN on $12{,}000$ fixed
  random collocation points (\ref{line:rr_A12k}), and FT-PINN
  (\ref{line:rr_B}).}
  \label{fig:rr_slices}
\end{figure}

\section{Conclusion}
\label{sec:conclude}
% !TEX root = main.tex
We have presented a feature-tracking PINN for conservation laws whose
solutions contain shocks. The solution network is defined on a fixed reference
domain and composed with a deformation map represented as a mixture of radial
basis functions, whose centers, bandwidths, and vector weights are trained
jointly with the network weights against the pulled-back conservation-law
residual. A uniform reference collocation grid thereby induces a physical
collocation distribution that concentrates on the feature, with no a priori
knowledge of its location, no residual-based resampling, and no per-point
optimization; the entire evolution of the collocation set is encoded by a few
hundred parameters of the deformation map. Unlike the analytical maps of our
earlier work~\cite{THAKUR2026110419}, the RBF-based deformation can represent features of arbitrary geometry.
However, invertibility and boundary preservation, which are properties of the
analytical family, are handled by a one-sided fold barrier on $\det G$, which
penalizes folding without opposing the metric collapse that produces the
clustering, and by the tangential projector $I-\mathbf{n}\mathbf{n}^\top$
applied to the raw displacement, which enforces
$\mathcal{G}(\partial\Omega_0)=\partial\Omega_0$ exactly while leaving
boundary points free to slide along the boundary.

We considered four shock-dominated problems, which include space-time viscous Burgers with
two merging shocks, a decelerating Burgers shock, the one-dimensional Euler
shock tube in space-time, and steady two-dimensional Euler regular shock
reflection. In every case, FT-PINN resolved the feature on the same collocation
budget, architecture, and training horizon on which a vanilla PINN failed to do so. In the current work, 
we only considered rectangular domains and two-dimensional problems. In future work, we want to extend the
framework to geometrically complex domains and to three-dimensional problems.

\section*{Acknowledgments}
This work is supported by AFOSR award number FA9550-22-1-0004, ONR award number
N00014-22-1-2299, and NSF award number CBET-2338843.
The content of this publication does not necessarily reflect the position
or policy of any of these supporters, and no official endorsement should
be inferred.

\bibliographystyle{elsarticle-num-names}
\bibliography{biblio}

\appendix

\section{Exact solution of the decelerating-shock Burgers benchmark}
\label{app:burgdec}
In this appendix we derive in closed form the entropy solution of the
decelerating-shock benchmark of Section~\ref{subsec:numexp:burgdec}, together
with the parametric shock path, and verify that the path satisfies the
Rankine--Hugoniot condition identically.

\subsection{Setup}
\label{subsec:app:burgdec:setup}

Consider the inviscid Burgers equation in conservative form,
\begin{equation}
    \partial_t q + \partial_x f(q) = 0,
    \qquad f(q) = \tfrac{1}{2}q^2,
    \qquad (x_1,t) \in [-1,1]\times[0,1],
    \label{eqn:app-burgers}
\end{equation}
with the sharp ($\varepsilon\to0$) limit of the initial
condition~\eqref{eqn:burgdec-ic},
\begin{equation}
    q_0(x_1) =
    \begin{cases}
        2\,(x_1+1)^2, & -1 \le x_1 < 0,\\
        0, & 0 \le x_1 \le 1,
    \end{cases}
    \label{eqn:app-ic}
\end{equation}
a quadratic ramp increasing from $0$ at $x_1=-1$ to $2$ at $x_1=0^-$, followed by a
downward jump to $0$. Since $q=0$ at both endpoints for all times considered,
the flux $f(q)=q^2/2$ vanishes on the spatial boundary and the total mass
\begin{equation}
    M \;=\; \int_{-1}^{1} q(x_1,t)\,\mathrm{d}x_1
    \;=\; \int_{-1}^{0} 2\,(x_1+1)^2\,\mathrm{d}x_1
    \;=\; \frac{2}{3}
    \label{eqn:app-mass}
\end{equation}
is conserved.

\subsection{Characteristics and the smooth branch}
\label{subsec:app:burgdec:char}

\begin{lemma}[Characteristic structure]
\label{lem:app-char}
For each $\xi \in [-1,1]$, let the foot $(\xi,0)$ denote the point in
the initial plane from which a characteristic emanates. The characteristic
through the foot $(\xi,0)$ is the straight line
$x_1 = \xi + q_0(\xi)\,t$, along which $q \equiv q_0(\xi)$. For $t \ge 0$,
characteristics with feet $\xi \in (-1,0)$ do not intersect one another, and
characteristics with feet $\xi \in (0,1)$ are vertical lines carrying $q=0$.
Consequently, the only multivaluedness originates at the jump $x_1=0$, where a
shock forms immediately at $t=0$.
\end{lemma}

\begin{proof}
Along $\dot x_1 = q$ one has $\partial_t q + q\,\partial_{x_1} q = 0$, so $q$ is constant on
characteristics and each is a straight line of slope $\mathrm{d}x_1/\mathrm{d}t =
q_0(\xi)$. On $(-1,0)$ the map $\xi \mapsto \xi + q_0(\xi)\,t$ has derivative
$1 + q_0'(\xi)\,t = 1 + 4(\xi+1)\,t > 0$ for $t \ge 0$, since $q_0' = 4(\xi+1)
> 0$ there. The map is therefore strictly increasing in $\xi$, hence injective,
and no two left characteristics cross. Right of the jump $q_0 \equiv 0$ gives
vertical characteristics. At $x_1=0$ the left limit $q_0(0^-)=2$ exceeds the
right limit $0$, so the fastest left characteristic overtakes the vertical ones
instantly, and a shock is born at $(x_1,t)=(0,0)$.
\end{proof}

\begin{lemma}[Closed form left of the shock]
\label{lem:app-closed}
Let $x_s(t)$ denote the shock position. For $x_1 < x_s(t)$, the entropy solution
of \eqref{eqn:app-burgers}--\eqref{eqn:app-ic} is, with $s = x_1+1$,
\begin{equation}
    q(x_1,t) \;=\; \frac{4s^2}{4st + 1 + \sqrt{1+8st}},
    \label{eqn:app-exact}
\end{equation}
and $q(x_1,t)=0$ for $x_1 \ge x_s(t)$.
\end{lemma}

\begin{proof}
By Lemma~\ref{lem:app-char}, through each point $(x_1,t)$ left of the shock passes
exactly one characteristic. Eliminating its foot $\xi = x_1 - qt$ from the
parametric description yields the implicit relation $q = q_0(x_1-qt)$, i.e.
$q = 2\,(x_1 - qt + 1)^2$. Substituting $q = 2\sigma^2$ with $\sigma = \xi+1 \in
(0,1)$ and $s = x_1+1$ turns the characteristic relation $s = \sigma +
2\sigma^2 t$ into the quadratic
\begin{equation}
    2t\,\sigma^2 + \sigma - s = 0,
    \qquad
    \sigma = \frac{-1 \pm \sqrt{1+8st}}{4t}.
    \label{eqn:app-quad}
\end{equation}
Only the ``$+$'' root is admissible: the ``$-$'' root is negative (a foot
outside the ramp), while the ``$+$'' root lies in $(0,1)$ and satisfies
$\sigma \to s$ as $t \to 0^+$, matching the initial data. Then
$q = 2\sigma^2 = \bigl(\sqrt{1+8st}-1\bigr)^2 / (8t^2)$, and rationalizing via
$\sqrt{1+8st}-1 = 8st\big/\bigl(1+\sqrt{1+8st}\bigr)$ together with
$\bigl(1+\sqrt{1+8st}\bigr)^2 = 2\bigl(1+4st+\sqrt{1+8st}\bigr)$ gives
\eqref{eqn:app-exact}, a form free of cancellation as $t\to0$. Indeed
$q(x_1,0) = 4s^2/2 = 2(x_1+1)^2 = q_0(x_1)$.
\end{proof}

\subsection{The shock curve}
\label{subsec:app:burgdec:shock}

Across a discontinuity with left state $q_L$ and right state $q_R = 0$, the
Rankine--Hugoniot jump condition~\cite{LeVeque_2002} for
\eqref{eqn:app-burgers} with $f(q)=q^2/2$ gives the shock speed
\begin{equation}
    \dot x_s
    = \frac{f(q_L)-f(q_R)}{q_L - q_R}
    = \frac{q_L + q_R}{2}
    = \frac{q_L}{2}.
    \label{eqn:app-rh}
\end{equation}
Since $0 = q_R < \dot x_s = q_L/2 < q_L$, characteristics impinge on the shock
from both sides, so \eqref{eqn:app-rh} is the entropic (Lax-admissible) shock.

\begin{prop}[Parametric shock path]
\label{prop:app-shock}
Parametrize the left characteristics by $a = \xi + 1 \in (0,1]$, so that the
characteristic with foot $\xi = a-1$ carries the value $q_0(\xi) = 2a^2$. The
shock absorbs these characteristics in strictly decreasing order of $a$, and
the characteristic labeled $a$ is absorbed at
\begin{equation}
    t(a) = \frac{1-a^3}{3a^4},
    \qquad
    x_s(a) = (a-1) + \frac{2\,(1-a^3)}{3a^2},
    \label{eqn:app-shockpath}
\end{equation}
which together trace the shock curve; $a=1$ corresponds to $t=0$, $x_s=0$.
\end{prop}

\begin{proof}
\emph{Ordering.} By the jump condition~\eqref{eqn:app-rh} the shock speed equals half its left
state, so every surviving left characteristic (speed $2a^2$, twice the shock
speed at the moment of its absorption) overtakes the shock; those with larger
$a$ start closer to the shock and move faster, hence arrive earlier.
Absorption therefore proceeds from $a=1$ downward, and at each instant the left
state at the shock is $q_L = 2a^2$, where $a$ is the label currently arriving.

\emph{Shock on its newest characteristic.} At the moment of absorption of
label $a$, the shock point lies on that characteristic:
\begin{equation}
    x_s = (a-1) + 2a^2\,t.
    \label{eqn:app-onchar}
\end{equation}

\emph{Mass closes the system.} At that same moment the smooth region consists
of exactly the characteristics with labels $\sigma \in (0,a)$. Changing
variables from $x_1$ to $\sigma$ via $x_1 = (\sigma-1) + 2\sigma^2 t$, so
$\mathrm{d}x_1 = (1+4\sigma t)\,\mathrm{d}\sigma$, the conserved mass
\eqref{eqn:app-mass} reads
\begin{equation}
    \frac{2}{3}
    = \int_0^a 2\sigma^2\,(1+4\sigma t)\,\mathrm{d}\sigma
    = \frac{2a^3}{3} + 2t\,a^4 .
    \label{eqn:app-massparam}
\end{equation}
Solving \eqref{eqn:app-massparam} for $t$ gives $t(a)$ in
\eqref{eqn:app-shockpath}; substituting into \eqref{eqn:app-onchar} gives
$x_s(a)$.
\end{proof}

Differentiating \eqref{eqn:app-shockpath} gives
$\mathrm{d}x_s/\mathrm{d}t = a^2 = q_L/2$, so the curve satisfies
\eqref{eqn:app-rh} identically, as expected since the mass balance and the
jump condition are integral and differential forms of the same law.

\begin{remark}[Deceleration and confinement]
\label{rem:app-decel}
The shock speed $\dot x_s = a^2$ decreases strictly and monotonically along the
path: the shock is born at $x_s = 0$ with speed $1$ (label $a=1$, left state
$2$) and decelerates as its left state $q_L = 2a^2$ drops. At $t=1$ the
absorbed label solves $3a^4 + a^3 - 1 = 0$, i.e.\ $a \approx 0.688$, giving
$x_s \approx 0.636$ and speed $\approx 0.47$. The shock therefore remains
inside $[-1,1]$ throughout $[0,1]$. Moreover absorption never ceases, since
arriving characteristics always travel at exactly twice the instantaneous shock
speed. Asymptotically $t \sim 1/(3a^4)$ gives $q_L = 2a^2 \sim 2/\sqrt{3t}$,
so the shock slows algebraically without stopping.
\end{remark}

\section{Exact solution of the regular shock reflection benchmark}
\label{app:rr}
In this appendix we derive, by direct computation, the analytical three-state
solution~\eqref{eqn:rr-states} of the regular shock reflection benchmark of
Section~\ref{subsec:numexp:rr}. The configuration is sketched to scale in
Figure~\ref{fig:rr_schematic}: the incident shock enters at the inlet split
point $(0,x_{2,\mathrm{split}})$ at the angle $\beta$ to the freestream,
reaches the wall at $x_{1,\mathrm{imp}}$, and the reflected shock leaves the
impingement point at the angle $\phi_{\mathrm{refl}}$ to the wall. A
streamline is horizontal in state 1, is deflected by $\theta_w$ across the
incident leg, so that in state 2 it is inclined at $+\theta_w$ to the wall and
meets the reflected leg at the angle $\beta_{\mathrm{refl}}=\phi_{\mathrm{refl}}+\theta_w$,
and is turned back to wall-parallel in state 3. The
oblique-shock relations, themselves a consequence of the Rankine--Hugoniot
conditions for the steady Euler equations, are applied first across the
incident leg and then across the reflected leg, with the slip condition on the
top wall closing the system. The gas is ideal with $\gamma=1.4$, states are
reported as primitives $W=(\rho,u,v,P)$, the sound speed is
$a=\sqrt{\gamma P/\rho}$, and, as in Section~\ref{subsec:numexp:rr},
$\theta_w$ denotes the flow deflection angle, subscripted to distinguish it
from the network parameters $\theta$.

\begin{figure}[h]
  \centering
  \resizebox{0.92\textwidth}{!}{% Auto-generated by scripts/plot_rr_schematic.py -- do not edit by hand.
\begin{tikzpicture}[scale=6.5, font=\small,
  shock/.style={line width=1.4pt, black},
  streamline/.style={thin, blue!70!black},
  ref/.style={thin, densely dashed, gray!80!black},
  ang/.style={thin, black},
  state/.style={circle, draw, inner sep=1.2pt, font=\small\bfseries}]

% ---- domain
\draw[gray!60] (0,0) rectangle (2.0,1.0);
\draw[line width=1pt] (0,1.0) -- (2.0,1.0);
\node[above, font=\small] at (1.55,1.0) {slip wall, $x_2=1$};
\node[rotate=90, above, font=\small] at (0,0.22) {inlet, $x_1=0$};
\node[rotate=-90, above, font=\small] at (2.0,0.5) {outflow, $x_1=2$};
\node[below, font=\small] at (1.000,0) {$x_2=0$ (no condition)};
% axes
\draw[-{Stealth}] (-0.14,-0.16) -- (0.12,-0.16) node[right] {$x_1$};
\draw[-{Stealth}] (-0.14,-0.16) -- (-0.14,0.10) node[above] {$x_2$};

% ---- inlet split and impingement point
\node[left, font=\small] at (0.0000,0.5000) {$x_{2,\mathrm{split}}$};
\fill (0.0000,0.5000) circle (0.012);
\fill (0.9353,1.0000) circle (0.012);
\node[above, font=\small, yshift=1pt] at (0.9353,1.0) {$x_{1,\mathrm{imp}}$};

% ---- shocks
\draw[shock] (0.0000,0.5000) -- (0.9353,1.0000) -- (2.0000,0.5536);
\node[rotate=28.13, above, font=\small] at (0.4676,0.7500) {incident shock};
\node[rotate=-22.74, below, font=\small] at (1.5954,0.7233) {reflected shock};

% ---- streamline
\draw[streamline, -{Stealth[length=5pt]}] (0.0000,0.6600) -- (0.1496,0.6600);
\draw[streamline] (0.1496,0.6600) -- (0.2993,0.6600);
\draw[streamline, -{Stealth[length=5pt]}] (0.2993,0.6600) -- (0.8086,0.7498);
\draw[streamline] (0.8086,0.7498) -- (1.3179,0.8396);
\draw[streamline, -{Stealth[length=5pt]}] (1.3179,0.8396) -- (1.6589,0.8396);
\draw[streamline] (1.6589,0.8396) -- (2.0000,0.8396);

% ---- beta at the inlet split point (freestream direction is horizontal)
\draw[ref] (0.0000,0.5000) -- (0.5500,0.5000);
\draw[ang] (0.3000,0.5000) arc[start angle=0.000, end angle=28.130, radius=0.3000];
\node[font=\small] at (0.4074,0.6021) {$\beta$};

% ---- phi_refl at the impingement point (measured from the wall)
\draw[ang] (1.2488,0.8685) arc[start angle=-22.744, end angle=0.000, radius=0.3400];
\node[font=\small] at (1.3862,0.9093) {$\phi_{\mathrm{refl}}$};

% ---- theta_w and beta_refl where the streamline meets the reflected leg:
% dashed = wall-parallel reference through Q; the state-2 streamline arrives
% at +theta_w to it, and the reflected leg makes beta_refl with that stream.
\draw[ref] (1.3179,0.8396) -- (0.9179,0.8396);
\draw[ang] (1.1179,0.8396) arc[start angle=180.000, end angle=190.000, radius=0.2000];
\draw[ang] (1.0412,0.9556) arc[start angle=157.256, end angle=190.000, radius=0.3000];
\node[font=\small] at (0.9608,0.7096) {$\theta_w$};
\node[font=\small] at (0.9559,0.9165) {$\beta_{\mathrm{refl}}$};

% ---- states
\node[state] (s1) at (0.22,0.90) {1};
\node[right=2pt of s1, font=\small] {$W_1,\ M_1=2.90$};
\node[state] (s2) at (0.95,0.30) {2};
\node[right=2pt of s2, font=\small] {$W_2,\ M_2=2.42$};
\node[state] (s3) at (1.60,0.90) {3};
\node[right=2pt of s3, font=\small] {$W_3,\ M_3=2.02$};
\end{tikzpicture}}
  \caption{Regular shock reflection, drawn to scale on $\Omega=[0,2]\times[0,1]$:
  uniform states 1--3 separated by the incident and reflected shocks. The
  incident shock makes the angle $\beta$ with the horizontal freestream and
  impinges on the slip wall at $x_{1,\mathrm{imp}}$. The reflected shock is
  inclined at $\phi_{\mathrm{refl}}$ to the wall. A streamline (blue) is
  deflected by $\theta_w$ across the incident leg, so that it meets the
  reflected leg at $\beta_{\mathrm{refl}}=\phi_{\mathrm{refl}}+\theta_w$
  (dashed: wall-parallel reference), and is turned back to wall-parallel
  across it.}
  \label{fig:rr_schematic}
\end{figure}
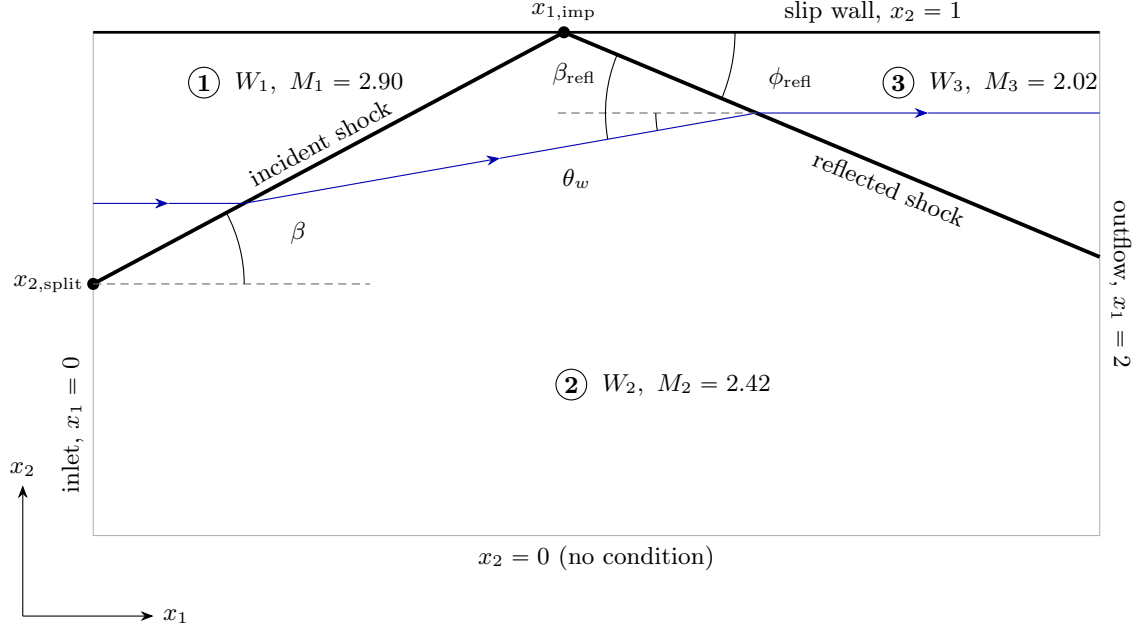

\subsection{Oblique-shock relations}
\label{subsec:app:rr:jump}

Consider a straight, stationary shock making an angle $\beta$ with an incoming
uniform stream of Mach number $M_1$. The Rankine--Hugoniot conditions across
the shock leave the tangential velocity component continuous and reduce, for
the normal components, to the one-dimensional normal-shock relations driven by
the normal Mach number $M_{n1} = M_1\sin\beta$, giving the density, pressure,
and normal-Mach jumps~\cite{anderson2003modern}
\begin{equation}
    \frac{\rho_2}{\rho_1}
    = \frac{(\gamma+1)M_{n1}^2}{(\gamma-1)M_{n1}^2 + 2},
    \qquad
    \frac{P_2}{P_1}
    = 1 + \frac{2\gamma}{\gamma+1}\bigl(M_{n1}^2 - 1\bigr),
    \qquad
    M_{n2}^2
    = \frac{(\gamma-1)M_{n1}^2 + 2}{2\gamma M_{n1}^2 - (\gamma-1)} .
    \label{eqn:app-rr-jumps}
\end{equation}
Because the normal velocity drops across the shock while the tangential one is
preserved, the velocity vector rotates toward the shock plane. The flow is
deflected by an angle $\theta_w$, and the downstream velocity makes the angle
$\beta-\theta_w$ with the shock, so that the downstream Mach number is
\begin{equation}
    M_2 = \frac{M_{n2}}{\sin(\beta-\theta_w)} .
    \label{eqn:app-rr-m2}
\end{equation}
Deflection and shock angle are tied by the $\theta$--$\beta$--$M$
relation~\cite{anderson2003modern}
\begin{equation}
    \tan\theta_w
    = 2\cot\beta\,
      \frac{M_1^2\sin^2\beta - 1}{M_1^2(\gamma + \cos 2\beta) + 2} .
    \label{eqn:app-rr-tbm}
\end{equation}
At fixed $M_1$, the deflection~\eqref{eqn:app-rr-tbm} viewed as a function of
$\beta$ rises from zero at the Mach angle $\beta=\arcsin(1/M_1)$ to a maximum
$\theta_{w,\max}(M_1)$ and falls back to zero at $\beta=90^\circ$. A
prescribed deflection $\theta_w<\theta_{w,\max}$ is therefore produced by two
shock angles, and attached shocks realize the smaller of the two, the
weak root, taken throughout below. If instead
$\theta_w>\theta_{w,\max}$, no shock angle can produce the required turn and
the shock detaches. This detachment criterion is what bounds the
regular-reflection regime.

\subsection{The incident shock}
\label{subsec:app:rr:incident}

The freestream state is $W_1 = (1,\ 3.431,\ 0,\ 1)$, so the freestream speed is
$V_1 = u_1 = 3.431$, $a_1 = \sqrt{\gamma} = 1.183$, and
$M_1 = V_1/a_1 = 2.90$, and the imposed deflection
is $\theta_w = 10^\circ$. Solving~\eqref{eqn:app-rr-tbm} for the weak root
gives the incident shock angle
\begin{equation}
    \beta = 28.13^\circ,
    \qquad
    M_{n1} = M_1\sin\beta = 1.367 .
    \label{eqn:app-rr-beta1}
\end{equation}
The jump relations~\eqref{eqn:app-rr-jumps} then give
\begin{equation}
    \frac{\rho_2}{\rho_1} = 1.633,
    \qquad
    \frac{P_2}{P_1} = 2.014,
    \qquad
    M_{n2} = 0.754,
    \qquad
    M_2 = \frac{M_{n2}}{\sin(\beta-\theta_w)} = 2.423 .
    \label{eqn:app-rr-state2-ratios}
\end{equation}
With $a_2 = \sqrt{\gamma P_2/\rho_2} = 1.314$ the downstream speed is $V_2 =
M_2 a_2 = 3.184$, directed at $+\theta_w$ to the wall, so
\begin{equation}
    u_2 = V_2\cos\theta_w = 3.136,
    \qquad
    v_2 = V_2\sin\theta_w = 0.553,
    \label{eqn:app-rr-state2}
\end{equation}
which is the state $W_2$ of~\eqref{eqn:rr-states}. Geometrically, the incident
leg emanates from the inlet split point $(0,\,x_{2,\mathrm{split}})=(0,\,0.5)$
at the angle $\beta$ to the freestream and meets the top wall $x_2 = 1$ at
\begin{equation}
    x_{1,\mathrm{imp}}
    = \frac{1 - x_{2,\mathrm{split}}}{\tan\beta}
    = \frac{0.5}{\tan 28.13^\circ}
    = 0.935 .
    \label{eqn:app-rr-imp}
\end{equation}

\subsection{The reflected shock and the three-state field}
\label{subsec:app:rr:reflected}

The state-2 stream is inclined at $+\theta_w$ to the top wall, which is an
inviscid slip wall. The flow must therefore be turned by $\theta_w$ back to
wall-parallel state, and the only steady supersonic mechanism available is a second
oblique shock anchored at the impingement point. Whether an attached solution
exists is decided at the post-incident Mach number. Evaluating the
maximum of~\eqref{eqn:app-rr-tbm} at $M_2 = 2.423$ gives
$\theta_{w,\max}(M_2) = 28.95^\circ$, and since $\theta_w = 10^\circ$ is well
below this detachment value, an attached weak shock exists and the reflection
is regular.

Applying~\eqref{eqn:app-rr-tbm} a second time at $(M_2,\theta_w)$, the weak
root is
\begin{equation}
    \beta_{\mathrm{refl}} = 32.74^\circ,
    \label{eqn:app-rr-beta2}
\end{equation}
measured from the state-2 flow direction. Since that direction is itself
inclined at $+\theta_w$ to the wall, the reflected leg is inclined at
\begin{equation}
    \phi_{\mathrm{refl}} = \beta_{\mathrm{refl}} - \theta_w = 22.74^\circ
    \label{eqn:app-rr-phi}
\end{equation}
to the wall. Note that $\phi_{\mathrm{refl}} \ne \beta$. The reflection is not
specular, because the incident and reflected shocks operate at different
upstream Mach numbers ($M_1=2.90$ versus $M_2=2.423$).

The jump relations~\eqref{eqn:app-rr-jumps}, now driven by the normal Mach
number $M_2\sin\beta_{\mathrm{refl}} = 1.310$, give
\begin{equation}
    \frac{\rho_3}{\rho_2} = 1.534,
    \qquad
    \frac{P_3}{P_2} = 1.837,
    \qquad
    M_3 = 2.019,
    \label{eqn:app-rr-state3-ratios}
\end{equation}
hence $\rho_3 = 2.505$ and $P_3 = 3.700$. With $a_3 = \sqrt{\gamma
P_3/\rho_3} = 1.438$ the downstream speed is $V_3 = M_3 a_3 = 2.904$, and the
turn restores wall-parallel flow,
\begin{equation}
    u_3 = V_3 = 2.904,
    \qquad
    v_3 = 0,
    \label{eqn:app-rr-state3}
\end{equation}
which is the state $W_3$ of~\eqref{eqn:rr-states}; $v_3=0$ confirms that the
slip condition on the top wall is met exactly.

\end{document}